\documentclass[12pt,draftcls,journal,onecolumn]{IEEEtran}
\usepackage{amsfonts,color,morefloats,pslatex,graphicx,graphics}
\usepackage{amssymb,amsthm, amsmath,latexsym}

\newtheorem{theorem}{Theorem}
\newtheorem{lemma}[theorem]{Lemma}
\newtheorem{remark}[theorem]{Remark}
\newtheorem{proposition}[theorem]{Proposition}

\newcommand{\ord}{{\mathrm{ord}}}

\newcommand{\lcm}{{\mathrm{lcm}}}

\newcommand{\gf}{{\mathrm{GF}}}

\newcommand{\Z}{\mathbb{{Z}}}

\newcommand{\cD}{{\mathcal{D}}}

\newcommand{\C}{{\mathcal{C}}}

\usepackage{blindtext}

\ifCLASSINFOpdf

\else

\fi

\begin{document}
%
\title{Improvement of the square-root low bounds on the minimum distances of BCH codes and Matrix-product codes
\thanks{
The work  was sponsored by the National Natural Science Foundation of China under Grant Numbers  62372247, 62272148, 12441103, in part by the Natural Science Foundation of Hubei Province of China
under Grant Number 2023AFB847,  by the Sunrise Program of Wuhan under the Grant Number 2023010201020419,  by the open research fund of National Mobile Communications Research Laboratory, Southeast University under Grant Number  2025D01, and by the National College Student Innovation Project under Grant Number 202410512027.
}
}
\author{Xiaoqiang Wang,\,\, Liuyi Li, \,\,  Yansheng Wu, \,\,Dabin Zheng, \,\,Shuxian Lu {\thanks{X. Wang, L. Li, D. Zheng and S. Lu are all with the Hubei Key Laboratory of Applied Mathematics, and with Key Laboratory of Intelligent Sensing System and Security, Ministry of Education, Faculty of Mathematics and Statistics, Hubei University, Wuhan, 430062, China. Email: waxiqq@163.com, llyhzq@163.com, dzheng@hubu.edu.cn, lulushuxian@163.com. }

\and\thanks{Y. Wu is  with the School of Computer Science, Nanjing University of Posts and Telecommunications, Nanjing
210023, China and the National Mobile Communications Research Laboratory, Southeast University, Nanjing, 211111, China.
 Email: yanshengwu@njupt.edu.cn. (Corresponding author) }

}

}

\maketitle

\begin{abstract}
 The task of constructing infinite families of self-dual codes with unbounded lengths and minimum distances exhibiting square-root lower bounds is extremely challenging, especially when it comes to cyclic codes.
Recently, the first infinite family of Euclidean self-dual binary and nonbinary cyclic codes, whose minimum distances have a square-root lower bound and have a lower bound better than square-root lower bounds are constructed in \cite{Chen23} for the lengths of these codes being unbounded.
Let $q$ be a power of a prime number and $Q=q^2$. In this paper, we first improve the lower bounds  on the minimum distances of Euclidean and Hermitian duals of BCH codes with length $\frac{q^m-1}{q^s-1}$ over $\mathbb{F}_q$ and $\frac{Q^m-1}{Q-1}$ over $\mathbb{F}_Q$ in \cite{Fan23,GDL21,Wang24} for the designed distances in some ranges, respectively, where $\frac{m}{s}\geq 3$.
Then based on matrix-product construction and some lower bounds on the minimum distances of BCH codes and their duals,
we obtain several classes of Euclidean and Hermitian self-dual codes, whose minimum distances have square-root lower bounds or a square-root-like lower bounds. Our lower bounds on the minimum distances of Euclidean and Hermitian self-dual cyclic codes improved many results in \cite{Chen23}. In addition, our lower bounds on the minimum distances of the duals of BCH codes are almost $q^s-1$ or $q$ times that of the existing lower bounds.

\end{abstract}



%

\section{Introduction}\label{sec-intro}

Let $q$ be a power of  some prime number $p$.
An $[n,k,d]$ linear code $\C$ over $\mathbb{F}_q$ is a $k$-dimensional subspace of $\mathbb{F}_q^n$ with minimum
Hamming distance $d$.  Let
$$\mathcal{C}^{\perp}=\{\mathbf{b} \in \mathbb{F}_q^n\,:\,\mathbf{b}\mathbf{c}^{T}=0 \,\,\text{for any $\mathbf{c} \in \mathcal{C}$}\}, $$
where $\mathbf{b}\mathbf{c}^{T}$ is the standard inner product of two vectors $\mathbf{b}$ and $\mathbf{c}$.
Then $\mathcal{C}^{\perp}$ is called {\it the dual} of $\mathcal{C}$.
A linear code $\mathcal{C}\subset \mathbb{F}_q^n$ is {\it self-orthogonal} if $\mathcal{C}\subseteq \mathcal{C}^{\perp}$, {\it dual-containing}  if $\mathcal{C}^{\perp}\subseteq\mathcal{C}$, and {\it self-dual} if $\mathcal{C}= \mathcal{C}^{\perp}$. A lower bound on the minimum distance of an $[n, k, d]$ code is called the {\it square-root lower bound} if it is
$\sqrt{n}$, and the {\it square-root-like lower bound} if it is $c\sqrt{n}$ for a fixed positive constant $c$ and fixed $q$. {\it In this paper, for square-root-like lower bound, we only consider the case that $c=\frac{1}{2}$. }

Cyclic codes is an important  class of linear codes.
 An $[n,k]$ linear code $\C$ over $\mathbb{F}_q$ is called {\em cyclic} if
$(c_0,c_1, \cdots, c_{n-1}) \in \C$ implies $(c_{n-1}, c_0, c_1, \cdots, c_{n-2})
\in \C$.
If we identify any vector $(c_0,c_1, \cdots, c_{n-1}) \in \mathbb{F}_q^n$
with  the polynomial
$$
\sum_{i=0}^{n-1} c_i x^i  \in \mathbb{F}_q[x]/\langle x^n-1 \rangle,
$$
any code $\C$ of length $n$ over $\mathbb{F}_q$ corresponds to a subset of the quotient ring
$\mathbb{F}_q[x]/\langle x^n-1 \rangle$.
A linear code $\C$ is cyclic if and only if the corresponding subset in $\mathbb{F}_q[x]/\langle x^n-1 \rangle$
is an ideal of the ring $\mathbb{F}_q[x]/\langle x^n-1 \rangle$.  It is well known that every ideal of $\mathbb{F}_q[x]/\langle x^n-1 \rangle$ is principal. Let $\C=\langle g(x) \rangle$ be a
cyclic code, where $g(x)$ is monic and has the smallest degree among all the
generators of $\C$. Then $g(x)$ is unique and called the {\em generator polynomial},
and $h(x)=(x^n-1)/g(x)$ is referred to as the {\em check polynomial} of $\C$.
The \emph{complement code}, denoted by $\C^c$, of the cyclic code $\C$ is the cyclic code
generated by $h(x)$. Hence, $\C + \C^c = \mathbb{F}_q^n$. The dual code $\C^\perp$ of $\C$ is generated
by the reciprocal polynomial of $h(x)$.

BCH codes is a class of cyclic codes, which has wide applications. Assume that $\gcd(n,q)=1$. Let $\ell=\ord_{n}(q)$ be the order of $q$ modulo $n$
and $\alpha$ be a generator of the group $\mathbb{F}_{q^\ell}^*$.
 For any $i$ with $0\leq i\leq q^\ell-2$,
let $m_i(x)$ denote the {\it minimal polynomial} of $\beta^i$ over $\mathbb{F}_{q}$, where $\beta=\alpha^{(q^\ell-1)/n}$ is a primitive $n$-th root of unity.
Let $\C_{(\delta,b)}$ be the cyclic code with generator polynomial
\begin{eqnarray*}
g_{(\delta,b)}(x)=\lcm \left(m_{b}(x), m_{b+1}(x), \ldots, m_{b+\delta-2}(x)\right),
\end{eqnarray*}
where $2\leq \delta\leq n$, $b$ is an integer and lcm denotes the least common multiple of these minimal polynomials. The code  $\C_{(\delta,b)}$
is termed a BCH code over $\mathbb{F}_q$ with length $n$ and designed distance $\delta$. Specifically, when $b=1$, the code $\C_{(\delta,b)}$ is referred to as a narrow-sense BCH code and we will subsequently abbreviate it as $\C_{\delta}$.

Binary BCH codes were discovered by Bose, Ray-Chaudhuri, and  Hocquenghem around 1960 in \cite{Bose62,Hocquenghem59}.  Then Gorenstein and Zierler generalized these codes to all finite fields \cite{Gorenstein61}. Over the past few decades, BCH codes have been extensively researched.
However, the knowledge about the minimum distances of the duals of BCH codes is quite limited. To the best of our knowledge, known results about the minimum distances of the duals of BCH codes are the following.

\begin{itemize}
\item In \cite{MS78}, the authors showed the classical Sidel'nikov bound, and the Carlitz-Uchiyama bound on the minimum distances of Euclidean duals of binary primitive BCH codes with odd designed distances.

\item  In \cite{Augot96}, the authors presented the lower bounds on the minimum distances of Euclidean duals of primitive BCH codes via the adaptation of the Weil bound to cyclic codes.

\item   In \cite{Fu24,GDL21,Wang24}, the authors gave the lower bounds of the minimum distances of Euclidean duals of BCH codes with some special lengths.

\item  In \cite{Fan23}, the authors gave the lower bounds on the minimum distances of Hermitian  duals of BCH codes with length $Q^m-1$ and $\frac{Q^m-1}{Q-1}$, where $Q=q^2$.
\end{itemize}

Let $\delta_0$ be the designed distance of narrow-sense BCH code with length $\frac{q^m-1}{q^s-1}$, where $\frac{m}{s}\geq 3$. The authors in \cite{GDL21,Wang24} showed the lower bounds on the minimum distances of Euclidean duals of BCH codes for $\frac{q^{st}-1}{q^s-1}< \delta_0 \leq \frac{q^{st}-1}{q^s-1}+q^{t-1}$, where $2\leq t\leq \frac{m}{s}-1$. Similarly, let $\delta_1$ be the designed distance of narrow-sense BCH code with length $\frac{Q^m-1}{Q-1}$, where $m\geq 3$. The authors in \cite{Fan23}
showed the lower bounds on the minimum distances the minimum distances of Hermitian duals of BCH codes for $\frac{Q^t-1}{Q-1}< \delta_1 \leq \frac{Q^t-1}{Q-1}+q^{t-1}$, where $2\leq t\leq m-1$. {\it In this paper, by considering different sets of zeroes in the duals of the above BCH codes from those considered in \cite{GDL21,Wang24,Fan23},  we show some new bounds on the minimum distances of Euclidean and Hermitian duals of above BCH codes. Compared with the known bounds,  the values of our lower bounds on the minimum distances of the duals of BCH codes are almost $q^s-1$ or $Q-1$ times the existing lower bounds. }


Constructing
Euclidean self-dual codes have been a hot topic for many years, and significant progress has been made in the study of this type of codes. In the finite length systems, a lot of papers have been published that focus on constructing self-dual GRS codes (refer to \cite{Fang19,Grass08,Jin17,Zhang20}, and the references). For a binary self-dual code of length $n$, the minimum distance $d$ satisfies $ d\leq 4\left\lfloor \frac{n}{24}\right\rfloor+6$ if $n\equiv22 \pmod {24}$ and $d\leq 4\left\lfloor \frac{n}{24}\right\rfloor+4$ for the other cases. A self-dual code that achieves these upper bounds in the binary case is called an extremal self-dual code (\cite{MS78}, Chapter 19). For a ternary self-dual code of length $n$, the minimum distance $d$ satisfies $ d\leq 3\left\lfloor \frac{n}{12}\right\rfloor+3$. A ternary self-dual code is called an
extremal self-dual ternary code if this code attaining the above bound.
 The
construction of extremal or optimal self-dual codes or self-dual codes of large minimum distances over small fields.
The readers can refer to \cite{Betsumiya23,Dougherty97,Gaborit03,Gulliver08,Harada07,Shi18}, and the references. For self-dual codes over small fields with the largest known minimum distances, the reader is referred to [20].
Until to now, there are a few results about infinite families of Euclidean and Hermitian self-dual codes of unbounded length $n$ over any finite fields with minimum distance $d\geq \sqrt{n}$. The reader is referred to, for example, \cite{Ding22,Huffman03,Pless72,Sun24} for information.

The matrix-product codes $C := [C_1 , C_2 , \ldots , C_s ]A$ over finite fields, introduced by Blackmore and Norton \cite{BN}, are a useful type of linear codes with larger lengths by combining several commensurate linear codes $C_1 , C_2 , \ldots , C_s$ of the same length with a defining matrix $A$. The said construction can be viewed as a generalization of the well-known $(u | u + v)$-construction and $(u+v+w | 2u+v | u)$-construction (see \cite{BN}). There are many papers focusing on the study of Hamming distance and decoding of matrix product codes, for example, \cite{VanAschB2008, FanY2014, Hernando2010} for information.

 Recently, by using special matrix-product construction,  namely  $(u | u + v)$-construction, the first infinite family of Euclidean self-dual binary and nonbinary cyclic codes  with and minimum distances have the square-root lower bound and have a lower bound better than the square-root lower bound are constructed for the lengths of the codes being unbounded in \cite{Chen23}.
{\it In this paper,
based on matrix-product construction and some lower bounds on the minimum distances of BCH codes and their duals,
 we obtain several classes of Euclidean and Hermitian self-dual codes, whose minimum distances have a square-root lower bound and a square-root-like lower bound. We generalized some results  in \cite{Chen23}. Moreover, for the same codes, our lower bounds are larger than the lower bounds in \cite{Chen23} in some cases. }

The rest of this paper is organized as follows. Section II contains some preliminaries. Sections III show some new bounds on the minimum distances of Euclidean and Hermitian duals of BCH codes with length $\frac{q^m-1}{q^s-1}$ and $\frac{Q^m-1}{Q-1}$, respectively. Section IV construct several infinite families of Euclidean and Hermitian self-dual linear codes with minimum distances better than  square-root lower bounds and square-root-like lower bounds. Section V concludes the paper.

\section{Preliminaries}

In this section, we present some
basic concepts and results that will be utilized subsequently. Unless stated otherwise, from this point forward, we will use the notation outlined below:
\begin{itemize}
\item $\mathbb{F}_q$ and $\mathbb{F}_Q$ are the finite fields with $q$ and $Q$ elements, where $q$ is a prime power and $Q=q^2$.
\item $\alpha$ is a primitive element of $\mathbb{F}_{q^m}$ and $\beta=\alpha^{\frac{q^m-1}{n}}$ is a primitive $n$-th root of unity, where $n \, |\, q^m-1$.
\item $\mathcal{C}_{\delta}$ denotes the narrow-sense BCH code with designed distance $\delta$, generator polynomial $g_{(\delta,1)}$ and length $n$.
\item $T=\{0\leq i\leq n-1: g_{(\delta,1)}(\beta^i)=0\}$ is the defining set of $\mathcal{C}_{\delta}$ with respect to $\beta$.
\item $T^{-1}=\{n-i\,:\,i\in T\}$ and $T^{-q}=\{(n-qi) \,\,{\rm mod} \,\,n\,:\,i\in T\}$.
\item  $\mathcal{C}^{\perp}_{\delta}$ and $\mathcal{C}^{\perp H}_{\delta}$ are Euclidean and Hermitian duals of $\mathcal{C}_{\delta}$, respectively.
\item $T^{\perp}$ and $T^{\perp H}$ are the defining sets of  $\mathcal{C}^{\perp}_{\delta}$ and $\mathcal{C}^{\perp H}_{\delta}$, respectively.
\item ${\rm CL}(a)$ denotes the $q$-cyclotomic coset leader modulo $n$ containing $a$, where $a$ is a positive integer with $1\leq a< n$.
\item $d^{\perp}(\mathcal{C}_{\delta})$ and $d^{\perp H}(\mathcal{C}_{\delta})$ denote the minimum distances of the duals of  $\mathcal{C}^{\perp}_{\delta}$ and  $\mathcal{C}^{\perp H}_{\delta}$, respectively.
\end{itemize}

\subsection{Cyclotomic coset and  $q$-adic expansion}

Let $n$ be an integer and $\Z_n$ be the ring of integers modulo $n$. The $q$-cyclotomic coset $C_i^{(q,n)}$
modulo $n$ containing $i$ is defined by
$$
C_i^{(q,n)}=\{ iq^j \bmod n: 0 \leq j \leq \ell_i-1\},
$$
where $\ell_i$ is the least positive integer such that $i \equiv i q^{\ell_i} \pmod{n}$, and is the size of $C_i^{(q,n)}$.

Let $m=\ord_n(q)$, i.e., the order of $q$ modulo $n$. Let $\beta$ be an $n$-th primitive root of unity in
$\gf(q^m)$.  Let $\C$ be a binary cyclic code of length $n$ with generator polynomial $g(x)$ and define
$$
T=\{0 \leq i \leq n-1: g(\beta^i)=0\}.
$$
Then the set $T$ is called the \emph{defining set} of $\C$ with respect to the $n$-th primitive root of unity $\beta$. By
definition, $T$ must be the union of some $q$-cyclotomic cosets modulo $n$. The smallest integer in $C_i^{(q,n)}$ is called the {\it coset leader} of $C_i^{(q,n)}$.

Let $i$  be an integer with $0<i<q^m-1$, then the $q$-adic expansion of $i$ can be written as
 $$i=(i_{m-1},i_{m-2},\ldots,i_{0})_q,$$ where
  $i=i_{m-1}q^{m-1}+i_{m-2}q^{m-2}+\cdots+i_1q+i_0$.
Let $0\leq j\leq m-1$, for any $iq^j\pmod {q^m-1}$, by definition, the $q$-adic expansion can be expressed as
$$iq^j\,\,\, (\text{mod} \,\,\,q^m-1)= (i_{m-j-1},i_{m-j-2},\ldots, i_{m-j})_q,$$
which is called the {\it circular $j$-left-shift} of $(i_{m-1},i_{m-2},\ldots,i_0)_q$, where the subscript of each coordinate is regarded as an integer modulo $m$.
The following results are known in \cite{Wang24}.

\begin{lemma}\cite[Lemma 3]{Wang24}\label{lem1b21}
Let $0<a,\,b\leq q^m-1$ be two positive integers with $q$-adic expansion $$a=(a_{m-1},a_{m-2},\ldots,a_0)_q\,\,\,\text{and}\,\,\,b=(b_{m-1},b_{m-2},\ldots,b_0)_q.$$
Then coset leader of  $\mathbb{C}_a$ modulo $q^m-1$ is greater than or equal to $b$ if and only if the circular $j$-left-shift of $(a_{m-1},a_{m-2},\ldots,a_0)$ is greater than or equal to
$(b_{m-1},b_{m-2},\ldots,b_0)$ for each $0 \leq j \leq m-1$.
\end{lemma}

\begin{lemma}\cite[Lemma 4]{Wang24}\label{lem:0913}
Let $0<t<q^m-1$ be a positive integer. Let $\mu$ be a common factor of $t$ and $q^m-1$, then $t$ is a coset leader modulo $q^m-1$ if and only if $\frac{t}{\mu}$ is a coset leader modulo $\frac{q^m-1}{\mu}$.
\end{lemma}

\subsection{Matrix-product codes and self-dual codes}

For $A = (a_{ij} ) \in  \mathbb M_{s\times t} (\mathbb{F}_q)$, if the rows of $A$ are linearly independent, then we say that $A$ is a full-row-rank (FRR) matrix. In particular, for $A = (a_{ij} ) \in \mathbb M_{s\times s} (\mathbb{F}_q )$, if $A$ is FRR, then we say that $A$ is a non-singular matrix.

Suppose that $C_i$ has generator matrix $G_i$ with $1\le i\le s$ and $A = (a_{ij} ) \in  M_{s\times t} (\mathbb{F}_q )$. Then  $[C_1 , C_2 , \ldots , C_s ]A$  is a classical code over $\mathbb{F}_q$ generated by the following matrix

\begin{equation*} G=
\left( \begin{array}{cccccc}
 a_{11}G_1& a_{12}G_1& \ldots & a_{1t}G_1\\
 a_{21}G_2& a_{22}G_2& \ldots & a_{2t}G_2\\
\vdots& \vdots& \ldots & \vdots & \\
 a_{s1}G_s& a_{s2}G_s& \ldots & a_{st}G_s\\
\end{array} \right).
\end{equation*}

Let us denote by $R_i=(a_{i1}, \ldots, a_{it})$ the element of $\mathbb{F}_q^t$ consisting of the $i$-th row of the matrix $A$ for $i=1, \ldots, s$. We denote by $U_A(k)$ the classical code generated by $\langle R_1, \ldots, R_k\rangle $ in $\mathbb{F}_q^t$, where $1\le k\le s$.

The following lemma play an important role in proving our main result.

\begin{lemma}\label{lem:mpcode}\cite[Theorem 3.7]{BN}
Let $A$ be an $s\times t$ an FRR matrix and $C_1, \ldots, C_s$ be classical codes over $\mathbb{F}_q$ with parameters $[n,k_i,d_i]$. Let $C=[C_1 , C_2 , \ldots , C_s ]A$. Then $C$ is an $[nt, \sum_{i=1}^sk_i, d(C)]$ classical code over $\mathbb{F}_q$. Moreover, we have $$d(C)\ge \mbox{min} \{d_1\cdot d(U_A(1)), \ldots, d_s\cdot d(U_A(s))\}$$ and the equality holds when $A$ is also triangle.
\end{lemma}

Let $\mathbf{u}=(u_1,u_2,\ldots,u_{n})\in \mathbb{F}_q^n$ and $\mathbf{v}=(v_1,v_2,\ldots,v_{n})\in \mathbb{F}_q^n$
be two vectors of length $n$, their {\it Euclidean inner product} and {\it Hermitian inner product} are defined as
$$\langle \mathbf u,\mathbf v\rangle =\sum_{i=1}^{n}u_iv_j,\,\,\,\text{and}\,\,\,\langle \mathbf u,\mathbf v\rangle_H =\sum_{i=1}^{n}u_iv_j^l,$$
respectively, where $q=l^2$ for some prime power $l$.

Let $Q=q^2$. The {\it Euclidean dual code} $C^\perp$ of $C$ and the {\it  Hermitian dual code} $C^{\perp_H}$ of $C$ are defined as
$$C^\perp =\{ \mathbf x\in  \mathbb{F}_q^{n}\, |\, \langle\mathbf x,\mathbf c\rangle =0,\forall \mathbf c\in C\}\,\, \text{and}\,\,C^{\perp_H} =\{ \mathbf x\in  \mathbb{F}_Q^n\, |\, \langle\mathbf x,\mathbf c\rangle_H =0,\forall \mathbf c\in C\},$$
 respectively.
If $C=C^\perp$, then $C$ is called {\it Euclidean self-dual}. If $C=C^{\perp_H}$, then $C$ is called {\it Hermitian self-dual}. The following lemma is given in \cite{Auly} and will be needed later.

\begin{lemma}\cite[Theorem 1]{Auly}\label{lem:1122}
Let $\mathbf{C}$ be a narrow-sense BCH code of length $n$ over $\mathbb{F}_q$ with designed distance $\delta$. Suppose that $m=ord_n(q)$,
then $\mathbf{C}$ is dual-containing if the designed distance $\delta$ is in the range
\begin{equation}\label{eq:1209}
2\leq \delta \leq \frac{n}{q^m-1}(q^{\lceil m/2\rceil}-1-(q-2)\,[m \,\,odd]),
\end{equation}
 where $[m \,\,odd]=1$ if $m$ is odd and $[m \,\,odd]=0$ if $m$ is even.
 \end{lemma}

%

\section{The lower bounds on the minimum distances of the Euclidean (Hermitian) duals of BCH codes}

 Let $q$ be a prime power and $Q=q^2$. In \cite{GDL21} and \cite{Wang24}, the authors developed the lower bounds on the minimum distances of Euclidean duals of narrow-sense BCH codes with length  $\frac{q^m-1}{q^s-1}$ over $\mathbb{F}_q$ for $s=1$ and $s\, | \,m$, respectively. In \cite{Fan23},  the authors showed the lower bounds on the minimum distances of Hermitian duals of narrow-sense BCH codes with length  $\frac{Q^m-1}{Q-1}$ over $\mathbb{F}_Q$. The propose of this section is to improve  the lower bounds on the minimum distances of Euclidean  and Hermitian duals of the above codes when the designed distance in some ranges.
 These results will be used in construction of self-dual codes with square-root-like lower bounds in later section.  We start with the following lemma.

\begin{lemma}\label{eq:0406}
 Let $s,t,m$ be positive integers such that $2\leq t\leq \frac{m}{s}-1$. Then the coset leader of $\mathbb{C}_{q^m-2q^{m-st+s}+q^{m-st}+u(q^s-1)}$ modulo $q^m-1$ is greater than $q^{st}+q^{s+t-1}-q^{t-1}-1$, where $q\neq 2$ and $1\leq u \leq q^{m-st}-1$.
\end{lemma}

\begin{proof}
Since $2\leq t\leq \frac{m}{s}-1$, we know that $m-st\leq s$.
The proof for the cases $m-st=s$ and $m-st<s$ is distinct. Hence, the proof should be treated separately.

\noindent {\bf Case 1:} $m-st=s$. In this case $1\leq u \leq q^s-1$ and  the $q$-adic expansion of $u$ can be wrote as
\begin{equation}\label{eq:0828}
u=(a_{s-1},a_{s-2},\ldots,a_0)_q.
\end{equation}
Let $l$ be the smallest positive integer in (\ref{eq:0828}) such that  $a_l> 0$.
Then the $q$-adic expansion of $u(q^s-1)$ can be expressed as
\begin{equation}\label{eq:s11s}
\begin{split}
u(q^s-1)=\underbrace{a_{s-1},\ldots,a_{l+1}}_{s-l-1},a_l-1,\underbrace{q-1,\ldots,q-1}_{l},\underbrace{q-1-a_{s-1},\ldots,q-1-a_{l+1}}_{s-l-1},q-a_l,\underbrace{0,\ldots,0}_{l})_q.
\end{split}
\end{equation}
From (\ref{eq:s11s}) the $q$-adic expansion of $q^m-2q^{2s}+q^{s}+u(q^s-1)$ can be expressed as
\begin{equation*}
\begin{split}
q^m-2q^{2s}+q^{s}+u(q^s-1)=&((\underbrace{q-1,\ldots,q-1}_{m-2s-1},q-2,\underbrace{a_{s-1},\ldots,a_{l+1}}_{s-l-1},a_l,\underbrace{0,\ldots,0}_{l},\\
&\underbrace{q-1-a_{s-1},\ldots,q-1-a_{l+1}}_{s-l-1},q-a_l,\underbrace{0,\ldots,0}_{l})_q).
\end{split}
\end{equation*}
Then  the coset leader of $\mathbb{C}_{q^m-2q^{m-st+s}+q^{m-st}+u(q^s-1)}$ modulo $q^m-1$ is greater than or equal to
\begin{equation}\label{eq:031128}
{\rm CL}((\underbrace{q-1,\ldots,q-1}_{m-2s-1},q-2,\underbrace{0,\ldots,0,}_{s-l-1}a_l,\underbrace{0,\ldots,0}_{s-1},q-a_l,\underbrace{0,\ldots,0}_{l})_q).
\end{equation}
Obviously, if ${\rm CL}((\underbrace{q-1,\ldots,q-1}_{m-2s-1},q-2,\underbrace{0,\ldots,0,}_{s-l-1}a_l,\underbrace{0,\ldots,0}_{s-1},q-a_l,\underbrace{0,\ldots,0}_{l})_q)$ achieves its minimum value, we have $l=0$ or $l=s-1$.
It is easy to check that (\ref{eq:031128}) can be written as
\begin{equation}\label{eq:082901}
\begin{split}
(\underbrace{0,\ldots,0,}_{s-1}a_l,\underbrace{0,\ldots,0}_{s-1},q-a_l,\underbrace{q-1,\ldots,q-1}_{m-2s-1},q-2)_q\\
\end{split}
\end{equation}
and
\begin{equation}\label{eq:1082902}
\begin{split}
(\underbrace{0,\ldots,0}_{s-1},q-a_l,\underbrace{0,\ldots,0}_{s-1},\underbrace{q-1,\ldots,q-1}_{m-2s-1},q-2,a_l)_q
\end{split}
\end{equation}
if $l=0$ and $l=s-1$, respectively. From the expression of  (\ref{eq:082901}) and (\ref{eq:1082902}), we know that (\ref{eq:082901}) and (\ref{eq:1082902})  achieve thier minimum values if and only if $a_l=1$ and $a_l=q-1$, respectively, which implies that the minimum values in  (\ref{eq:082901}) and (\ref{eq:1082902}) are $q^{st}+q^{st-s+1}-2$ and  $q^{st}+q^{st-s+1}-q- 1$.

Hence,
the coset leader of $\mathbb{C}_{q^m-2q^{m-st+s}+q^{m-st}+u(q^s-1)}$ modulo $q^m-1$ is greater than $q^{st}+q^{s+t-1}-q^{t-1}-1$ if $m-st=s$.

\noindent {\bf Case 2:} $m-st>s$. If $u=1$, it is clear that the $q$-adic expansion of $q^m-2q^{m-st+s}+q^{m-st}+q^s-1$ can be expressed as
$$(\underbrace{q-1,\ldots,q-1}_{st-s-1},q-2,\underbrace{0,\ldots,0}_{s-1},1,\underbrace{0,\ldots,0}_{m-st-s},\underbrace{q-1,\ldots,q-1}_s)_q.$$
Then the coset leader of $\mathbb{C}_{q^m-2q^{m-st+s}+q^{m-st}+q^s-1}$ modulo $q^m-1$ is greater than
$$(\underbrace{0,\ldots,0}_{m-st-1},1,\underbrace{q-1,\ldots,q-1}_{st-1},q-2)=2q^{st}-2> q^{st}+q^{s+t-1}-q^{t-1}-1.$$
Hence, the result is hold for $u=1$.

In the following, we consider the case that $u\geq 2$. From the range of $u$, we know that $u(q^s-1)< q^{m-st+s}-q^{m-st}$.
Then the $q$-adic expansion of $u(q^s-1)$ can be expressed as
\begin{equation}\label{eq:s1s}
u(q^s-1)=(i_{m-st+s-1},i_{m-st+s-2},\ldots,i_0),
\end{equation}
where $i_{m-st+s-1}<q-1$. Let $l$ be the largest positive integer in (\ref{eq:s1s}) such that  $i_l> 0$. Since $u\geq 2$, we know that $l\geq s$.
From (\ref{eq:s1s}) we have
\begin{equation}\label{eq:10815}
 (q^s-1) \, | \, i_{l}q^{l}+i_{l-1}q^{l-1}+\cdots+i_1q+i_0.
 \end{equation}
Since
$i_lq^l\equiv i_lq^{l-s}\pmod {q^s-1},$
we know that $i_{l}\neq 0$ and $i_{l-s}\neq 0$ in (\ref{eq:10815}). Obviously, $l-s\neq m-st$ since $l\leq m-st+s-1$.
We prove the coset leader of $\mathbb{C}_{q^m-2q^{m-st+s}+q^{m-st}+u(q^s-1)}$ modulo $q^m-1$ is greater than $q^{st}+q^{s+t-1}-q^{t-1}-1$ from the following two subcases.

\noindent {\bf Subcase 2.1:} $l< m-st$. From (\ref{eq:s1s}), the $q$-adic expansion of $q^m-2q^{m-st+s}+q^{m-st}+u(q^s-1)$ can be expressed as
\begin{equation*}
\begin{split}
&q^m-2q^{m-st+s}+q^{m-st}+u(q^s-1)\\
=&(\underbrace{q-1,\ldots,q-1}_{st-s-1},q-2,\underbrace{0,\ldots,0,}_{s-1}1,\underbrace{0,\ldots,0}_{m-st-l-1}, i_{l},\underbrace{i_{l-1},\ldots,i_{l-s+1}}_{s-1}, i_{l-s},\underbrace{i_{l-s-1},\ldots,i_0}_{l-s})_q.
\end{split}
\end{equation*}
Let $T_0=max\{s-1,l-s,m-st-l-1\}$. From the range of $l$, we know that $T_0\leq m-st-2$.
Then we have
\begin{equation*}
\begin{split}
&{\rm CL}((\underbrace{q-1,\ldots,q-1}_{st-s-1},q-2,\underbrace{0,\ldots,0,}_{s-1}1,\underbrace{0,\cdots,0}_{m-st-l-1}, i_{l},\underbrace{i_{l-1},\ldots,i_{l-s+1}}_{s-1}, i_{l-s},\underbrace{i_{l-s-1},\ldots,i_0}_{l-s})_q)\\
\geq&{\rm CL}((\underbrace{q-1,\ldots,q-1}_{st-s-1},q-2,\underbrace{0,\ldots,0,}_{s-1}1,\underbrace{0,\ldots,0}_{m-st-l-1}, i_{l},\underbrace{0,\ldots,0}_{s-1}, i_{l-s},\underbrace{0,\ldots,0}_{l-s})_q)\\
\geq&{\rm CL}((\underbrace{q-1,\ldots,q-1}_{st-s-1},q-2,\underbrace{0,\ldots,0,}_{s-1}1,\underbrace{0,\ldots,0}_{m-st-l-1}, 1,\underbrace{0,\ldots,0}_{s-1}, 1,\underbrace{0,\ldots,0}_{l-s})_q)\\
\geq&(\underbrace{0,\ldots,0}_{T_0},1,\ldots)_q
\geq q^{st+1}.
\end{split}
\end{equation*}
\noindent {\bf Subcase 2.2:} $l\geq m-st$. From $i_{m-st+s-1}<q-1$ and (\ref{eq:s1s}), we obtain that the $q$-adic expansion of $q^{m-st}+u(q^s-1)$ can be expressed as
\begin{equation}\label{eq:mst}
q^{m-st}+u(q^s-1)=(a_{m-st+s-1},\ldots,a_{m-st}, i_{m-st-1},\ldots,i_0),
\end{equation}
which implies that
\begin{equation*}
\begin{split}
&q^m-2q^{m-st+s}+q^{m-st}+u(q^s-1)\\
=&(\underbrace{q-1,\ldots,q-1}_{st-s-1},q-2,a_{m-st+s-1},\ldots,a_{m-st}, i_{m-st-1},\ldots,i_0)_q.
\end{split}
\end{equation*}
 Let $h$ be the largest positive integer in (\ref{eq:mst}) such that  $a_h> 0$ and $T_1=max\{l-s,h-l+s-1,m-st+s-h-1\}$. By the definition of $h$ and $l$, we know that $m-st+s-1\geq h\geq l$, then  $h-l+s-1<2s-1$. Since $m-st\neq s$ and $2\leq t\leq \frac{m}{s}-1$, we obtain that $m-st\geq 2s$, which implies that $T_1\leq m-st-2$. Then
 \begin{equation*}
\begin{split}
&{\rm CL}((\underbrace{q-1,\ldots,q-1}_{st-s-1},q-2,a_{m-st+s-1},\ldots,a_{m-st}, i_{m-st-1},\ldots,i_0)_q)\\
\geq&{\rm CL}((\underbrace{q-1,\ldots,q-1}_{st-s-1},q-2,\underbrace{0,\ldots,0,}_{m-st+s-h-1}a_h,\underbrace{0,\ldots,0}_{h-l+s-1}, i_{l-s},\underbrace{0,\ldots,0}_{l-s})_q)\\
\geq&{\rm CL}((\underbrace{q-1,\ldots,q-1}_{st-s-1},q-2,\underbrace{0,\ldots,0,}_{m-st+s-h-1}1,\underbrace{0,\ldots,0}_{h-l+s-1}, 1,\underbrace{0,\ldots,0}_{l-s})_q)\\
\geq&(\underbrace{0,\ldots,0}_{T_1},1,\ldots)_q
\geq q^{st+1}.
\end{split}
\end{equation*}
Hence, ${\rm CL}((\underbrace{q-1,\ldots,q-1}_{st-s-1},q-2,a_{m-st+s-1},\ldots,a_{m-st}, i_{m-st-1},\ldots,i_0)_q)$  is greater than $q^{st}+q^{s+t-1}-q^{t-1}-1$ if $m-st>s$.

 Therefore, combining the above two cases, the desired conclusion follows.
\end{proof}

\begin{theorem}\label{eq:1204}
Let $\C_{\delta}$ be the narrow-sense BCH code over $\mathbb{F}_q$ with length $n=\frac{q^m-1}{q^s-1}$ and designed distance $\delta$, where $q\neq 2$ be a prime power. Let $2\leq t\leq \frac{m}{s}-1$ and $\frac{q^{st}-1}{q^s-1}< \delta \leq \frac{q^{st}-1}{q^s-1}+q^{t-1}$,  then $d(\C_{\delta}^{\perp})\geq q^{m-st}$.
\end{theorem}
\begin{proof}
From the BCH bound, in order to obtain the desired result, we only need to show that $\{\frac{q^{m-(t-1)s}-1}{q^s-1}+1,\frac{q^{m-(t-1)s}-1}{q^s-1}+2,
\ldots,\frac{q^{m-(t-1)s}-1}{q^s-1}+q^{m-st}-1 \}\subseteq T^{\perp}$. Let $1\leq u \leq q^{m-st}-1$ and $i=\frac{q^{m-(t-1)s}-1}{q^s-1}+q^{m-st}-u$, then
\begin{equation*}
\begin{split}
n-i = \frac{q^m-2q^{m-st+s}+q^{m-st}+u(q^s-1)}{q^s-1}.
\end{split}
\end{equation*}
Let ${\rm CL}(n-i)=A$, then the the coset leader of $\mathbb{C}_{q^m-2q^{m-t+1}+q^{m-t}+u(q^s-1)}$ is $A(q^s-1)$ from Lemma \ref{lem:0913}. Hence, from Lemma \ref{eq:0406} we know that
$$A>\left\lceil \frac{q^{st}+q^{s+t-1}-q^{t-1}-1}{q^s-1} \right\rceil,$$
which implies that
\begin{equation}\label{eq:1207}
{\rm CL}(n-i)={\rm CL}\left(\frac{q^m-2q^{m-t+1}+q^{m-t}+u(q-1)}{q^s-1}\right) > \delta-1.
\end{equation}
By the definition of $T$ and $T^{-1}$,
we obtain that $n-i \notin T$ and $i \notin T^{-1}$ in (\ref{eq:1207}). From the definition of $T^{\perp}$, we have $T^{\perp}=\mathbb{Z}_n \setminus T^{-1}$, which implies that $i=\frac{q^{m-(t-1)s}-1}{q^s-1}+q^{m-st}-u\in T^{\perp}$ for any $1\leq u \leq q^{m-st}-1$.  The desired conclusion follows.
\end{proof}

\begin{remark}\label{rem:0416}
Let $\frac{q^{st}-1}{q^s-1}< \delta \leq \frac{q^{st}-1}{q^s-1}+q^{t-1}$, where $2\leq t\leq \frac{m}{s}-1$. The authors in \cite{GDL21} showed that $d(\C_{\delta}^{\perp})\geq\frac{q^{m-t}-1}{q-1}+1$ for $s=1$. Then the authors in \cite{Wang24} generalized their results from the case $s=1$ to the case $s\, |\, m$ and proved that  $d(\C_{\delta}^{\perp})\geq \frac{q^{m-ts}-1}{q^s-1}+1$ or $d(\C_{\delta}^{\perp})\geq \frac{q^{m-ts}-1}{q^s-1}+2$ for $s\, |\, m$. In Theorem \ref{eq:1204}, we improved these bounds and showed that $d^{\perp}(\delta)\geq q^{m-st}$ when the designed distance in these ranges. Our lower bounds on the minimum distances of the duals of BCH codes are almost $q^s-1$ times that of the existing lower bounds. We show some examples in Table 1.
\end{remark}
\begin{center}$${{\rm Table\,\,\, 1:\,\,\,Lower\,\,\, bounds\,\,\, on\,\,\, the\,\,\, minimum\,\,\, distance\,\,\, of\,\,\, \mathcal{C}^{\perp}_{\delta}\,\,\, for } \,\,\,s=1}$$
\begin{tabular}{|c|c|c|c|c|c|c|c|}
\hline
 $t$&$m$&$q$&$\delta$&$d(\C_{\delta}^{\perp})\geq$&$d(\C_{\delta}^{\perp})\geq$&$d(\C_{\delta}^{\perp})\geq$\\ \cline{5-7}
& & &  &Ref. \cite{GDL21} &Ref. \cite{Wang24}&Theorem \ref{eq:1204}  \\
\hline
2&4&3&5-7&5&7&9\\
\hline
2&5&3&5-7&14&16&27\\
\hline
3&5&3&14-22&5&7&9\\
\hline
2&4&5&7-11&7&9&25\\
\hline
\end{tabular}
\end{center}

In the following, our task is to give a new lower bound on the minimum distance of $\mathcal{C}^{\perp H}_{\delta}$ for $\delta$ in some ranges. The following lemma is needed.

\begin{lemma}\label{eq:0412}
 Let $Q=q^2$ be a prime power, $2\leq t\leq m-1$ and $0\leq u \leq Q^{m-t}$.  Then the coset leader of
 $\mathbb{C}_{Q^m-2Q^{m-t+1}q+Q^{m-t}q+u(Q-1)q+q-1}$ modulo $Q^m-1$ is greater than $Q^t+Qq^{t-1}-q^{t-1}-1$.
\end{lemma}

\begin{proof}
If $u=0$, it is easy to see that the $Q$-adic expansion of $Q^m-2Q^{m-t+1}q+Q^{m-t}q+u(Q-1)q+q-1$ can be expressed as
$$(\underbrace{Q-1,\ldots,Q-1,}_{t-2}Q-2q,q,\underbrace{0,\ldots,0,}_{m-t-1}q-1)_Q.$$
Then the coset leader of
 $\mathbb{C}_{Q^m-2Q^{m-t+1}q+Q^{m-t}q+u(Q-1)q+q-1}$ modulo $Q^m-1$ is
\begin{equation*}
\begin{split}
(\underbrace{0,\ldots,0,}_{m-t-1}q-1,\underbrace{Q-1,\ldots,Q-1,}_{t-2}Q-2q,q)_Q,
\end{split}
\end{equation*}
which is greater than $Q^t+Qq^{t-1}-q^{t-1}-1$.

If $u=qQ^l$ for $0\leq l\leq m-t-1$, then the $Q$-adic expansion of $u(Q-1)q$ contains one non-zero terms and the $Q$-adic expansion of $Q^m-2Q^{m-t+1}q+Q^{m-t}q+u(Q-1)q+q-1$ can be expressed as
$$(\underbrace{Q-1,\ldots,Q-1,}_{t-2}Q-2q,q,\underbrace{0,\ldots,0,}_{m-t-l-2}Q-1,\underbrace{0,\ldots,0,}_{l}q-1)_Q$$
if $0\leq l\leq m-t-2$ and
$$(\underbrace{Q-1,\ldots,Q-1,}_{t-2}Q-2q+1,q-1,\underbrace{0,\ldots,0,}_{m-t-1}q-1)_Q$$
if $l= m-t-1$. Hence, from Lemma \ref{lem1b21} we can obtain that  $\mathbb{C}_{Q^m-2Q^{m-t+1}q+Q^{m-t}q+u(Q-1)q+q-1}$ modulo $Q^m-1$ is greater than $Q^t+Qq^{t-1}-q^{t-1}-1$ for $u=qQ^l$.

In the following, we always assume that  $0< u \leq Q^{m-t}$ and $u\neq qQ^l$ for $0\leq l\leq m-t-1$.
Let the $Q$-adic expansion of $u(Q-1)q$ be expressed as
\begin{equation}\label{eq:120908}
u(Q-1)q=a_{m-t+1}Q^{m-t+1}+a_{m-t}Q^{m-t}+\cdots+a_0,
\end{equation}
where $a_{m-t+1}<q$. From $q\, | \, u(Q-1)q$, we know that $q\, |\, a_0$, which implies that $a_0+q-1<Q$.
If the $Q$-adic expansion of $u(Q-1)q$ contains one non-zero terms, then from (\ref{eq:120908}) we have
$u(Q-1)q=a_{l+1}Q^{l+1}$, which implies that $u=qQ^{l}$ and $a_{l+1}=Q-1$. It is contradictory to $u\neq qQ^l$ for $0\leq l\leq m-t-1$.
Hence, the $Q$-adic expansion of $u(Q-1)q$ contains at least two non-zero terms in (\ref{eq:120908}).

Assume that
the $Q$-adic expansion of $Q^{m-t}q+u(Q-1)q+q-1$ be expressed as
 \begin{equation}\label{eq:120909}
Q^{m-t}q+u(Q-1)q+q-1=b_{m-t+1}Q^{m-t+1}+b_{m-t}Q^{m-t}+\cdots+b_0.
\end{equation}
From (\ref{eq:120908}), we know that $b_0=a_0+q-1$ and $b_{m-t+1}\leq q$. Hence, the $Q$-adic expansion of $Q^m-2Q^{m-t+1}q+Q^{m-t}q+u(Q-1)q+q-1$ can be expressed as
$$(\underbrace{Q-1,\ldots,Q-1,}_{t-2}Q-2q+b_{m-t+1},b_{m-t},\ldots,b_0)_Q.$$

Since the $Q$-adic expansion of $u(Q-1)q$ contains at least two non-zero terms in (\ref{eq:120908}), it is easy to check that the $Q$-adic expansion of $Q^{m-t}q+u(Q-1)q+q-1$  contains at least two non-zero terms in (\ref{eq:120909}). Moreover,
the $Q$-adic expansion of $Q^{m-t}q+u(Q-1)q+q-1$  contains two non-zero terms if and only if $a_{m-t}=Q-q$, $a_{m-t+1}\neq 0$, $a_0\neq 0$ and $a_i=0$ for $1\leq i\leq m-t-1$ in (\ref{eq:120908}), i.e.,
\begin{equation*}
u(Q-1)q=a_{m-t+1}Q^{m-t+1}+(Q-q)Q^{m-t}+a_0,
\end{equation*}
which implies that
\begin{equation}\label{091001}
(Q-1)\,|\,a_{m-t+1}Q^{m-t+1}+(Q-q)Q^{m-t}+a_0.
\end{equation}
Since $aQ^t\equiv a \pmod {Q-1}$ for any positive integers $a$ and $t$, from (\ref{091001}) we have $(Q-1)\,|\,a_{m-t+1}+Q-q+a_0$, which contradicts with $a_{m-t+1}<q$ and $q\, | \,a_0$. This implies that the $Q$-adic expansion of $Q^{m-t}q+u(Q-1)q+q-1$ contains at least three non-zero terms in (\ref{eq:120909}). Hence,
 \begin{equation}\label{eq:0912}
\begin{split}
&{\rm CL}((\underbrace{Q-1,\ldots,Q-1,}_{t-2}Q-2q+b_{m-t+1},b_{m-t},\ldots,b_0)_Q)\\
\geq&{\rm CL}((\underbrace{0,\ldots,0}_{m-t-1},b_1,b_0,\underbrace{Q-1,\ldots,Q-1,}_{t-2}Q-2q+b_{m-t+1})_Q),
\end{split}
\end{equation}
where $b_0\neq 0$, $b_1\neq 0$ and $b_{m-t+1}\neq 0$.

 Since $b_0=a_0+q-1$ and $q\, |\, a_0$,
it is clear that ${\rm CL}((\underbrace{0,\ldots,0}_{m-t-1},b_1,b_0,\underbrace{Q-1,\ldots,Q-1,}_{t-2}Q-2q+b_{m-t+1})_Q)$ can achieve its minimum value if $b_0=q-1$ and $b_1=1$.
In this case, from (\ref{eq:120909}) we have
$$Q^{m-t}q+u(Q-1)q+q-1=b_{m-t+1}Q^{m-t+1}+Q+q-1,$$
which is the same as
$$uq=\frac{b_{m-t+1}Q^{m-t+1}+Q-Q^{m-t}q}{Q-1},$$
which implies that
$$(Q-1)\, | \, b_{m-t+1}+1-q.$$
Then $b_{m-t+1}=q-1$. Hence, from (\ref{eq:0912}) we obtain
\begin{equation*}
\begin{split}
{\rm CL}((\underbrace{Q-1,\ldots,Q-1,}_{t-2}Q-2q+b_{m-t+1},b_{m-t},\ldots,b_0)_Q)
&\geq Q^t+qQ^{t-1}-Q^{t-1}+Q-q+1\\
&\geq Q^t+Qq^{t-1}-q^{t-1}-1.
\end{split}
\end{equation*}
Therefore,  the coset leader of
 $\mathbb{C}_{Q^m-2Q^{m-t+1}q+Q^{m-t}q+u(Q-1)q+q-1}$ modulo $Q^m-1$ is greater than $Q^t+Qq^{t-1}-q^{t-1}-1$. This completes the proof.
\end{proof}

\begin{theorem}\label{eq:0822}
Let $Q=q^2$ and $\C_{\delta}$ be the narrow-sense BCH code over $\mathbb{F}_Q$ with length $n=\frac{Q^m-1}{Q-1}$ and designed distance $\delta$.
 Let $2\leq t\leq m-1$ and $\frac{Q^t-1}{Q-1}< \delta \leq \frac{Q^t-1}{Q-1}+q^{t-1}$,  then $d(\C_{\delta }^{\perp H})\geq Q^{m-t}+1$.
\end{theorem}
\begin{proof}
From the BCH bound, in order to obtain the desired result, we only need to show that $\{\frac{Q^{m-t+1}-1}{Q-1}+1,\frac{Q^{m-t+1}-1}{Q-1}+2,
\ldots,\frac{Q^{m-t+1}-1}{Q-1}+Q^{m-t} \}\subseteq T^{\perp H}$. Let $0\leq u \leq Q^{m-t}-1 $ and $i=\frac{Q^{m-t+1}-1}{Q-1}+Q^{m-t}-u$, we have that
\begin{equation*}
\begin{split}
n-qi = \frac{Q^m-2Q^{m-t+1}q+Q^{m-t}q+u(Q-1)q+q-1}{Q-1}.
\end{split}
\end{equation*}

Assume that ${\rm CL}(n-qi)=A$, from Lemma \ref{lem:0913} the  coset leader of $\mathbb{C}_{Q^m-2Q^{m-t+1}q+Q^{m-t}q+u(Q-1)q+q-1}$ is $A(Q-1)$. Then we know that
$$A>\left\lceil \frac{Q^{t}+Qq^{t-1}-q^{t-1}-1}{Q-1} \right\rceil$$
by Lemma \ref{eq:0412}, which implies that
\begin{equation}\label{eq:120701}
{\rm CL}(n-qi)={\rm CL}\left(\frac{Q^m-2Q^{m-t+1}q+Q^{m-t}q+u(Q-1)q+q-1}{Q-1}\right) > \delta-1.
\end{equation}
By the definition of $T$ and $T^{-q}$, we obtain that $n-qi \notin T$ and $i \notin T^{-q}$ in (\ref{eq:120701}). from the definition of
$T^{\perp H}$, we know that
$T^{\perp H}=\mathbb{Z}_n \setminus T^{-q}$, which implies that $i=\frac{q^{m-(t-1)s}-1}{q^s-1}+q^{m-st}-u\in T^{\perp H}$ for any $1\leq u \leq q^{m-st}-1$.  The desired conclusion follows.
\end{proof}

\begin{remark}\label{rem:0416}
Let $\frac{Q^t-1}{Q-1}< \delta \leq \frac{Q^t-1}{Q-1}+q^{t-1}$, where $2\leq t\leq m-1$. The authors in \cite{GDL21} showed that $d(\C_{\delta}^{\perp})\geq\frac{Q^{m-t}q-q}{Q-1}+1$. In Theorem \ref{eq:0822}, we improved the bound and showed that $d^{\perp}(\delta)\geq Q^{m-t}+1$ when the designed distance in the above ranges. Our lower bounds on the minimum distances of the duals of BCH codes are almost $q$ times that of the existing lower bounds. We show some examples in Table 2.
\end{remark}
\begin{center}$${{\rm Table\,\,\, 2:\,\,\,Lower\,\,\, bounds\,\,\, on\,\,\, the\,\,\, minimum\,\,\, distance\,\,\, of\,\,\, \mathcal{C}^{\perp H}_{\delta}\,\,\, for } \,\,\,s=1}$$
\begin{tabular}{|c|c|c|c|c|c|c|}
\hline
$t$&$m$&$Q$&$\delta$&$d^{\perp}(\C_{\delta }^{\perp H})\geq$&$d^{\perp}(\C_{\delta }^{\perp H})\geq$\\ \cline{5-6}
& &   & &Ref. \cite{GDL21}&Theorem \ref{eq:0822}  \\
\hline
2&4&9&11-13&31&81\\
\hline
2&4&25&27-29&131&625\\
\hline
\end{tabular}
\end{center}

\section{Self-dual  Euclidean  and Hermitian codes with minimum distance large than square-root or square-root-like lower bound}

The purpose of this section is to construct several classes Euclidean  and Hermitian self-dual codes with minimum distances better than square-root or square-root-like lower bounds by using BCH codes and matrix product codes. Firstly, we give several classes linear codes with dimensions equal to half of their lengths and minimum distances better than square-root or square-root-like lower bounds.

\subsection{Linear codes with minimum distances better than square-root or square-root-like lower bounds over $\mathbb{F}_q$. }

In this subsection, let $n=\frac{q^m-1}{\lambda}$, or $n=\frac{q^m-1}{q^s-1}$, or $n=q^m+1$, where $\lambda \, | \, q-1$ and $s\, |\,m$. Let $\mathcal{C}_\delta$ be a primitive narrow-sense BCH code with length $n$ and designed distance $\delta$ over  $\mathbb{F}_q$ and
$\mathcal{C}_\delta^{\perp}$ be its Euclidean dual. Let
\begin{equation}\label{eq:11125}
\mathcal{C'_\delta}=[\mathcal{C}_\delta,\mathcal{C}_\delta^{\perp}]A
\end{equation}
or
\begin{equation}\label{eq:1112501}
\mathcal{C'_\delta}=[\mathcal{C}_\delta^{\perp},\mathcal{C}_\delta]A,
\end{equation}
where $A$ is a $2\times 2$ nonsigular matrix.
From Lemma \ref{lem:mpcode},  $\mathcal{C'_\delta}$ is a linear code with parameters $[2n,n,d(\mathcal{C'_\delta})]$, where
\begin{equation}\label{eq:11127}
d(\mathcal{C'_\delta})\geq \min\{2d(\mathcal{C}^{\perp}_{\delta}),d(\mathcal{C}_{\delta})\}
 \end{equation}
 if $\mathcal{C'_\delta}$ is in (\ref{eq:11125}), and
\begin{equation}\label{eq:1112701}
d(\mathcal{C'_\delta})\geq \min\{2d(\mathcal{C}^{\perp}_{\delta}),d(\mathcal{C}_{\delta})\}
 \end{equation}
  if $\mathcal{C'_\delta}$ is in (\ref{eq:1112501}).

 In the following, we construct several classes of linear codes of length $2n$ with minimum distances better than  square-root lower bound ($d= \sqrt{2n}$) or square-root like lower bound ($d= \sqrt{\frac{n}{2}}$).
We first recall the following known result.

\begin{lemma}\label{eq:thm19190}
Let $1\leq s\leq \frac{q-1}{\lambda}-1$, $0\leq t\leq m-2$, $n=\frac{q^m-1}{\lambda}$, $\lambda \, | \, q-1$ and $\lambda\neq q-1$. Let $d(\mathcal{C}^{\perp}_{\delta})$ be the minimum distance of $\mathcal{C}^{\perp}_{\delta}$. Then we have
\begin{equation*}
\begin{aligned}
d(\mathcal{C}^{\perp}_{\delta})\geq
\begin{cases}
\frac{q^{m-t}+\lambda-1}{\lambda}-sq^{m-t-1}, &\text{if $\delta = \frac{q^{t+1}-q}{\lambda}+s+1$},\\
\frac{q^{m-t-1}-1}{\lambda}-s+3, &\text{if $\frac{q^{t+1}-1}{\lambda}+(s-1)q^{t+1} < \delta \leq \frac{q^{t+1}-1}{\lambda}+sq^{t+1}-1$},\\
\frac{q^{m-t-1}-1}{\lambda}-s+2, &\text{if $\delta = \frac{q^{t+1}-1}{\lambda}+sq^{t+1}$},\\
\frac{q^{m-t-1}-q}{\lambda}+2, &\text{if $\frac{q^{t+2}-1}{\lambda}-q^{t+1} < \delta \leq \frac{q^{t+2}-q}{\lambda}+1$}.
               \end{cases}
\end{aligned}
\end{equation*}
Moreover, if $\lambda=1$, $1 \leq s < q-2$ and $q\geq 5$, then $d^{\perp}(\delta)\geq q^{m-t-1}-s+3$ for $sq^{t+1} \leq\delta \leq (s+1)q^{t+1}-sq^t-2$.
\end{lemma}

From (\ref{eq:11125}), (\ref{eq:1112501}) and  Lemma \ref{eq:thm19190}, we can obtain obtain several classes linear codes with minimum distances better than square-root lower bounds, which are given in Theorem \ref{eq:thm9190} and Theorem \ref{eq:1thm91190}.

\begin{theorem}\label{eq:thm9190}
Let $1\leq s\leq \frac{q-1}{\lambda}-1$, $0\leq t\leq m-2$, $n=\frac{q^m-1}{\lambda}$, $\lambda \, | \, q-1$ and $\lambda\neq q-1$.  Let $\mathcal{C'_\delta}$ be given in (\ref{eq:11125}), then
\begin{equation}\label{eq:1113}
\begin{split}
d(\mathcal{C'_\delta})\geq \begin{cases}
\frac{q^{\frac{m+1}{2}}+\lambda-1}{\lambda}-sq^{\frac{m-1}{2}}, &\text{if $m$ is odd, $q\geq \sqrt{2\lambda q}+\lambda s$ and $\delta = \frac{q^{\frac{m+1}{2}}-q}{\lambda}+s+1$},\\
\min\{2\delta,d(\mathcal{C}^{\perp}_{\delta})\}, &\text{if $m$ is odd, $\lambda s \geq \sqrt{\frac{ \lambda q}{2}}$ and}\\
&\text{$\max\{\frac{\sqrt{2n}}{2},\frac{q^{\frac{m-1}{2}}-1}{\lambda}+(s-1) q^{\frac{m-1}{2}}\}< \delta \leq \frac{q^{\frac{m-1}{2}}-1}{\lambda}+s q^{\frac{m-1}{2}}-1$},\\
\frac{q^{\frac{m+1}{2}}-q}{\lambda}+2, &\text{if $m$ is odd and $\frac{q^{\frac{m+1}{2}}-1}{\lambda}-q^{\frac{m-1}{2}} < \delta \leq \frac{q^{\frac{m+1}{2}}-q}{\lambda}+1$},\\
2\delta, &\text{if $m$ is even and $ q^{\frac{m}{2}}-q^{\frac{m-2}{2}}-1<\delta \leq  q^{\frac{m}{2}}-q+1$},
\end{cases}
\end{split}
\end{equation}
where $d(\mathcal{C}^{\perp}_{\delta})=\frac{q^{\frac{m+1}{2}}-1}{\lambda}-s+3$ if $\delta\neq\frac{q^{\frac{m-1}{2}-1}}{\lambda}+s q^{\frac{m-1}{2}}$ and $d(\mathcal{C}^{\perp}_{\delta})=\frac{q^{\frac{m+1}{2}}-1}{\lambda}-s+2$ if $\delta=\frac{q^{\frac{m-1}{2}-1}}{\lambda}+s q^{\frac{m-1}{2}}$.
Moreover, if $m$ is odd, $q\geq 5$, $\lambda=1$ and $sq^{\frac{m-1}{2}} \leq\delta \leq (s+1)q^{\frac{m-1}{2}}-sq^{\frac{m-3}{2}}-2$, then
\begin{equation}\label{eq:111301}
\begin{split}
d(\mathcal{C'_\delta})\geq \begin{cases}
2\delta, &\text{if $\sqrt{\frac{q}{2}}\leq s\leq \frac{q}{2}$},\\
q^{\frac{m+1}{2}}-s+3, &\text{if $\frac{q}{2}< s< q-2$.}
\end{cases}
\end{split}
\end{equation}
All the given $d(\mathcal{C'_\delta})$ mentioned above are better than the square-root lower bounds.
\end{theorem}

\begin{proof} We only prove the results in (\ref{eq:1113}). The case in  (\ref{eq:111301}) is similar and we omit the details. There are four cases for discussions.

\noindent {\bf Case 1:} $m$ is odd and  $\delta = \frac{q^{\frac{m+1}{2}}-q}{\lambda}+s+1$.
From Lemma \ref{eq:thm19190} and the BCH bound, we know that
$$d(\mathcal{C}_{\delta})\geq \frac{q^{\frac{m+1}{2}}-q}{\lambda}+s+1\,\, \text{and}\,\, d(\mathcal{C}^{\perp}_{\delta})\geq \frac{q^{\frac{m+1}{2}}+\lambda-1}{\lambda}-sq^{\frac{m-1}{2}}.$$
It is easy to see that $2 (\frac{q^{\frac{m+1}{2}}-q}{\lambda}+s+1)> \frac{q^{\frac{m+1}{2}}+\lambda-1}{\lambda}-sq^{\frac{m-1}{2}}$, then from (\ref{eq:11127}) we have $d(\mathcal{C'_\delta})\geq \frac{q^{\frac{m+1}{2}}+\lambda-1}{\lambda}-sq^{\frac{m-1}{2}}$. If $q\geq \sqrt{2\lambda q}+\lambda s$, then
$$d(\mathcal{C'_\delta})\geq \frac{q^{\frac{m}{2}}(\sqrt{2\lambda})+\lambda-1}{\lambda}>\sqrt{2n}.$$

\noindent {\bf Case 2:} $m$ is odd, $\lambda s \geq \sqrt{\frac{ \lambda q}{2}}$ and $\max\{\frac{\sqrt{2n}}{2},\frac{q^{\frac{m-1}{2}}-1}{\lambda}+(s-1) q^{\frac{m-1}{2}}\}< \delta \leq \frac{q^{\frac{m-1}{2}}-1-\lambda}{\lambda}+s q^{\frac{m-1}{2}}$.
Since $\lambda s \geq \sqrt{\frac{ \lambda q}{2}}$, we know that
$$\frac{q^{\frac{m-1}{2}}-1-\lambda}{\lambda}+s q^{\frac{m-1}{2}}>\sqrt{\frac{q^m}{2}}>\frac{\sqrt{2n}}{2}.$$
Then the range of $\delta$ is not empty. From the BCH bound, we know that
\begin{equation}\label{112702}
2d(\mathcal{C}_{\delta})> 2\max\left\{ \frac{q^{\frac{m-1}{2}}-1}{\lambda}+(s-1) q^{\frac{m-1}{2}},\frac{\sqrt{2n}}{2}\right\}\geq\sqrt{2n}.
\end{equation}
Since $ s\leq \frac{q-1}{\lambda}-1$, we have $\lambda s \leq q-1-\lambda$. Then
\begin{equation}\label{112703}
 d(\mathcal{C}_{\delta}^{\perp})\geq \frac{q^{\frac{m+1}{2}}-1}{\lambda}-s+2>\sqrt{2n}.
 \end{equation}
From (\ref{eq:11127}), (\ref{112702}) and (\ref{112703}) we obtain that $d(\mathcal{C'_\delta})> \sqrt{2n}.$

\noindent {\bf Case 3:}  $m$ is odd and $\frac{q^{\frac{m+1}{2}}-1}{\lambda}-q^{\frac{m-1}{2}} < \delta \leq \frac{q^{\frac{m+1}{2}}-q}{\lambda}+1$.  From Lemma \ref{eq:thm19190} and the BCH bound, it is clear that
$$d(\mathcal{C}_{\delta})\geq \delta\,\, \text{and}\,\, d(\mathcal{C}_{\delta}^{\perp})\geq \frac{q^{\frac{m+1}{2}}-q}{\lambda}+2.$$
Obviously,  $2(\frac{q^{\frac{m+1}{2}}-1}{\lambda}-q^{\frac{m-1}{2}})\geq\frac{q^{\frac{m+1}{2}}-q}{\lambda}+2$ from the range of $\lambda$. From (\ref{eq:11127}) we have $d(\mathcal{C'_\delta})\geq \frac{q^{\frac{m+1}{2}}-q}{\lambda}+2$. If $2\lambda= q-1$, then
\begin{equation}\label{eq:112703}
(\frac{q^{\frac{m+1}{2}}-q}{\lambda}+2)^2=(\frac{q^{\frac{m+1}{2}}-1}{\lambda})^2=\frac{q^{m+1}-2q^{\frac{m+1}{2}}+1}{\lambda^2}>2n.
\end{equation}
If $2\lambda<q-1$, then $2\lambda\leq \lfloor \frac{2(q-1)}{3}\rfloor$ since $\lambda \, | \, q-1$.
Hence,
\begin{equation}\label{eq:112704}
(\frac{q^{\frac{m+1}{2}}-q}{\lambda}+2)^2> (\frac{q^{\frac{m+1}{2}}-q}{\lambda})^2=\frac{q^{m+1}-2q^{\frac{m+3}{2}}+q^2}{\lambda^2}\geq 2n.
\end{equation}
From (\ref{eq:112703}) and (\ref{eq:112704}),
we get that $d(\mathcal{C'_\delta})> \sqrt{2n}$.

\noindent {\bf Case 4:}  $m$ is even and $q^{\frac{m}{2}}-q^{\frac{m-2}{2}}-1<\delta \leq  q^{\frac{m}{2}}-q+1$.  From Lemma \ref{eq:thm19190} and the BCH bound, we know that
$$d(\mathcal{C}_{\delta})\geq q^{\frac{m}{2}}-q^{\frac{m-2}{2}}\,\, \text{and}\,\,d(\mathcal{C}_{\delta}^{\perp})\geq q^{\frac{m+2}{2}}-q+2.$$
Obviously, we have $2(q^{\frac{m}{2}}-q^{\frac{m-2}{2}})\leq q^{\frac{m+2}{2}}-q+2$, then from (\ref{eq:1127}) it is easy to see that
$$d(\mathcal{C'_\delta})\geq 2\delta\geq2(q^{\frac{m}{2}}-q^{\frac{m-2}{2}})> \sqrt{2(q^m-1)}.$$

Combining all the cases, the desired conclusion follows.
\end{proof}

\begin{theorem}\label{eq:1thm91190}
Let the symbols be given as in Theorem \ref{eq:thm9190} and $\mathcal{C'_\delta}$ be given in (\ref{eq:1112501}).
Then $\mathcal{C'_\delta}$ is a linear code with parameters $[2n,n,d(\mathcal{C'_\delta})]$, where
$$d(\mathcal{C'_\delta})\geq 2(\frac{q^{\frac{m+1}{2}}+\lambda-1}{\lambda}-sq^{\frac{m-1}{2}})$$
is $\delta = \frac{q^{\frac{m+1}{2}}-q}{\lambda}+s+1$ and $\lceil\frac{q}{2}\rceil\leq \lambda s\leq q-\sqrt{\lambda q/2}$, and
$$d(\mathcal{C'_\delta})\geq  \delta$$ if one of the following statements hold:
\begin{itemize}
\item[(1)] $\delta = \frac{q^{\frac{m+1}{2}}-q}{\lambda}+s+1$ and $\lambda s\leq \lfloor\frac{q-1}{2}\rfloor$;
\item[(2)]  $\lambda s-\lambda+1\geq \sqrt{2\lambda q}$ and $\frac{q^{\frac{m-1}{2}}-1}{\lambda}+(s-1)q^{\frac{m-1}{2}}<\delta<\frac{q^{\frac{m-1}{2}}-1-\lambda}{\lambda}+sq^{\frac{m-1}{2}}$;
\item[(3)] $q>\sqrt{4\lambda q+\lambda^2}$ and $\frac{q^{\frac{m+1}{2}}-1}{\lambda}-q^{\frac{m-1}{2}} < \delta \leq \frac{q^{\frac{m+1}{2}}-q}{\lambda}+1$;
\item[(4)]    $\lambda=1$ and  $sq^{\frac{m-1}{2}} \leq\delta \leq (s+1)q^{\frac{m-1}{2}}-sq^{\frac{m-3}{2}}-2$ for $\sqrt{2q}\leq s< q-2$, $q\geq 5$.
\end{itemize}
All the given $d(\mathcal{C'_\delta})$ mentioned above are better than the square-root lower bounds.
\end{theorem}

\begin{proof} From Lemma \ref{lem:mpcode}, we know that  $\C$ is a linear code with parameters $[2n,n]$. We only prove $d(\mathcal{C'_\delta})$ for the case that
$\delta = \frac{q^{\frac{m+1}{2}}-q}{\lambda}+s+1$. The other cases are similar and we omit the details.

Let $t=\frac{m-1}{2}$ in Lemma \ref{eq:thm19190} and from the BCH bound, we know that
$$d(\mathcal{C}_{\delta}^{\perp})\geq \frac{q^{\frac{m+1}{2}}+\lambda-1}{\lambda}-sq^{\frac{m-1}{2}}\,\, \text{and}\,\, d(\mathcal{C}_{\delta})\geq \frac{q^{\frac{m+1}{2}}-q}{\lambda}+s+1.$$
It is easy to check that $\min\{2(\frac{q^{\frac{m+1}{2}}+\lambda-1}{\lambda}-sq^{\frac{m-1}{2}}),\frac{q^{\frac{m+1}{2}}-q}{\lambda}+s+1\}> \sqrt{2n}$ if $\lambda s\leq q-\sqrt{\lambda q/2}$.
In addition,
$$2\left( q^{\frac{m+1}{2}}+\lambda-1-\lambda s q^{\frac{m-1}{2}} \right)\geq q^{\frac{m+1}{2}}-q+\lambda s+\lambda,$$
if and only if
$$ q^{\frac{m+1}{2}}\geq 2 \lambda s q^{\frac{m-1}{2}}-q+\lambda s+2-\lambda.$$
Hence, from (\ref{eq:1112701}) we can obtain that $d(\mathcal{C'_\delta})\geq \frac{q^{\frac{m+1}{2}}-q}{\lambda}+s+1$ and  $d(\mathcal{C'_\delta})\geq 2(\frac{q^{\frac{m+1}{2}}+\lambda-1}{\lambda}-sq^{\frac{m-1}{2}})$ if $\lambda s\leq \lfloor\frac{q-1}{2}\rfloor$ and $\lambda s\geq \lceil\frac{q}{2}\rceil$, respectively.  The desired conclusion follows.
\end{proof}

\begin{remark}\label{eq:1212}
From the conditions in Lemma \ref{eq:thm19190}, we know that $n\neq2^m-1$ in Lemma \ref{eq:thm19190}. In fact, the lower bounds on the minimum distances of the duals of narrow-sense BCH codes with designed distance $\delta$ and length $n=2^m-1$ have been given in \cite[Theorem 1]{Wang24}. By the same way as Theorem \ref{eq:thm9190} and Theorem \ref{eq:1thm91190}, we have the following result.

\begin{itemize}
\item[$\bullet$] Let $\mathcal{C'_\delta}$ be given in (\ref{eq:11125}),
then $\mathcal{C'_\delta}$ is a linear code with parameters $[2n, n, d(\mathcal{C'_\delta})]$,
then
\begin{equation*}
\begin{split}
d(\mathcal{C'_\delta}) \geq \begin{cases}
2\delta, & \text{if $m$ is even and $\frac{\sqrt{2^{m+1}-2}}{2} < \delta \leq 2^{\frac{m}{2}}-1$}, \\
2^{\frac{m+1}{2}}, & \text{if $m$ is odd and $2^{\frac{m-1}{2}} + 2^{\frac{m-3}{2}} - 1 < \delta \leq 2^{\frac{m+1}{2}} - 1$}, \\
2^{\frac{m+1}{2}}+1, & \text{if $m$ is odd and $2^{\frac{m-1}{2}}  < \delta \leq 2^{\frac{m-1}{2}} + 2^{\frac{m-3}{2}} - 1$},\\
2\delta,& \text{if $m$ is odd and $\delta=2^{\frac{m-1}{2}}$.}
\end{cases}
\end{split}
\end{equation*}

\item[$\bullet$] Let $\mathcal{C'_\delta}$ be given in (\ref{eq:1112501}), then $\C$ is a linear code with parameters $[2n, n, d(\mathcal{C'_\delta})]$, where
\begin{equation*}
\begin{split}
d(\mathcal{C'_\delta}) \geq \begin{cases}
\delta, & \text{if $m$ is even and ${\sqrt{2^{m+1}-2}} < \delta \leq 2^{\frac{m+2}{2}}-1$}, \\
2^{\frac{m+1}{2}}, & \text{if $m$ is odd and $2^{\frac{m+1}{2}} + 2^{\frac{m-1}{2}} - 1 < \delta \leq 2^{\frac{m+3}{2}}-1$}, \\
\delta, & \text{if $m$ is odd and $2^{\frac{m+1}{2}} \leq \delta \leq 2^{\frac{m+1}{2}} + 2$}, \\
2^{\frac{m+1}{2}}+2, & \text{if $m$ is odd and $2^{\frac{m+1}{2}} + 2 < \delta \leq 2^{\frac{m+1}{2}} + 2^{\frac{m-1}{2}} - 1$.}
\end{cases}
\end{split}
\end{equation*}
\end{itemize}
All the given $d(\mathcal{C'_\delta})$ mentioned above are better than the square-root lower bounds.
\end{remark}

\begin{theorem}\label{thm:16}
Let $q\geq 3$ be an odd prime power and $n=q^m+1$. Let $\mathcal{C'_\delta}$ be given in (\ref{eq:1112501}). If $m$ is even and $lq^{\frac{m}{2}}+ \frac{q^{\frac{m}{2}}+3}{2}  \leq \delta < (l+1)q^{\frac{m}{2}}+ \frac{q^{\frac{m}{2}}+3}{2}$ for $2\leq l \leq \frac{q-3}{2}$, then
$$d(\mathcal{C'_\delta})\geq 2q^{\frac{m}{2}}-4l-2,$$
which is better than square-root lower bound.
\end{theorem}

\begin{proof}
 If $m$ is even and $lq^{\frac{m}{2}}+ \frac{q^{\frac{m}{2}}+3}{2}  \leq \delta <  (l+1)q^{\frac{m}{2}}+ \frac{q^{\frac{m}{2}}+3}{2}$, from the BCH bound and \cite[Theorem 4.1]{Fu24}, we know that
$$d(\mathcal{C}_{\delta})\geq lq^{\frac{m}{2}}+\frac{q^{\frac{m}{2}}+3}{2}\,\, \text{and}\,\, d(\mathcal{C}_{\delta}^{\perp})\geq q^{\frac{m}{2}}-2l-1$$
for $2\leq l \leq \frac{q-3}{2}$. It is easy to see that $2(q^{\frac{m}{2}}-2l-1)< lq^{\frac{m}{2}}+\frac{q^{\frac{m}{2}}+3}{2}$, then from (\ref{eq:1112701}) we have
$$d(\mathcal{C'_\delta})\geq 2q^{\frac{m}{2}}-4l-2>\sqrt{2n}.$$
The desired conclusion follows.
\end{proof}

Below, we present two  classes of linear codes of length $2n=\frac{2(q^m-1)}{q^s-1}$ with dimension $n$ and minimum distances larger than like-square-root lower bound ($d=\sqrt{\frac{n}{2}}$). From Theorem \ref{eq:1204}, it is easy to get the following result. The proof is similar as above, we omit the details.

\begin{theorem}\label{thm:17}
Let $n=\frac{q^m-1}{q^s-1}$, $\frac{q^{st}-1}{q^s-1}< \delta \leq \frac{q^{st}-1}{q^s-1}+q^{t-1}$, $m+s=2st$ and $q\neq 2$. Then $\mathcal{\mathcal{C'_\delta}}$ is a linear code with parameters $[2n,n, d(\mathcal{C'_\delta})]$, where
\begin{equation*}
\begin{aligned}
d(\mathcal{C'_\delta})\geq\begin{cases}
q^{\frac{m-s}{2}}, &\text{ if  $\mathcal{C'_\delta}$ is given in (\ref{eq:11125}),} \\
\delta, &\text{if $\mathcal{C'_\delta}$ is given in (\ref{eq:1112501}),}
\end{cases}
\end{aligned}
\end{equation*}
which is better than like-square-root lower bound ($d> \sqrt{\frac{n}{2}}$).
\end{theorem}

\subsection{Linear codes with minimum distances large than square-root lower bound over $\mathbb{F}_Q$. }

In this subsection, let $n=Q^m-1$, or $n=\frac{Q^m-1}{Q-1}$, where $Q=q^2$. Let $\mathcal{C}_\delta$ be a primitive narrow-sense BCH code with length $n$ and designed distance $\delta$ over $\mathbb{F}_Q$, and
$\mathcal{C}_\delta^{\perp H}$ be its Hermitian dual. Let
\begin{equation}\label{eq:1125}
\mathcal{{C'_\delta}}=[\mathcal{C}_\delta, \mathcal{C}_\delta^{\perp H}]A
\end{equation}
or
\begin{equation}\label{eq:112501}
\mathcal{C'_\delta}=[\mathcal{C}_\delta^{\perp H},\mathcal{C}_\delta]A,
\end{equation}
where $A$ is a $2\times 2$ nonsigular matrix over $\mathbb{F}_Q$.
From Lemma \ref{lem:mpcode}, we obtain that $\mathcal{C'_\delta}^H$ is a linear code with parameters $[2n,n,d(\mathcal{C'_\delta})]$, where
\begin{equation}\label{eq:1127}
d(\mathcal{C'_\delta})\geq \min\{2d(\mathcal{C}),d(\mathcal{C}^{\perp H}_{\delta})\}
 \end{equation}
 if $\mathcal{C'_\delta}$ is in (\ref{eq:1125}), and
\begin{equation}\label{eq:112701}
d(\mathcal{C'_\delta})\geq \min\{2d(\mathcal{C}^{\perp H}_{\delta}),d(\mathcal{C}_{\delta})\}
 \end{equation}
  if $\mathcal{C'_\delta}$ is in (\ref{eq:112501}).
 In the following, we construct several classes of linear codes with minimum distances better than  square-root lower bound ($d= \sqrt{2n}$) or square-root like lower bound ($d= \sqrt{\frac{n}{2}}$) over $\mathbb{F}_Q$.
From \cite[Lemma 4]{Fan23} and the BCH bound, we have the following result.

\begin{lemma}\label{lem:19}
Let $n=Q^m-1$, $Q=q^2$ and $m\geq3$. Let $d(\mathcal{C}^{\perp H}_{\delta})$ be the minimum distance of $\mathcal{C}^{\perp H}_{\delta }$. Then we have
\begin{equation*}
\begin{aligned}
d(\mathcal{C}^{\perp H}_{\delta})\geq \begin{cases}
  (b+1)q^{2(m-t)+1}, &\text{if $\delta=q^{2t-1}-b$ $(2\leq t \leq m-1$, $1\leq b \leq q^2 -2)$},\\
q^{2(m-t)+1}-s+1, &\text{if  $sq^{2t-1}\leq\delta\leq(s+1)q^{2t-1}-1$ $(1\leq t\leq m$, $1\leq s \leq q-1)$}\\
q^{2(m-t)+1}-aq-s+1, &\text{if $(aq+s)q^{2t-1}\leq \delta \leq (aq+s+1)q^{2t-1}-1$,}\\
&\text{$(1\leq t \leq m-1$, $1\leq a \leq q-2$, $0\leq s \leq q-1)$},\\
q^{2(m-t)+1}-q^2+q-s+1, &\text{if $(q^2-q+s)q^{2t-1} \leq \delta \leq (q^2-q+s+1)q^{2t-1}-1$},\\
&(2\leq t \leq m-1$, $0\leq s \leq q-2),\\
q^{2(m-t)+1}-q^2+2, &\text{ if $(q^2-1)q^{2t-1} \leq \delta \leq q^{2t+1}-q^2+1$ \,\,$(2\leq t \leq m-1)$}.\\
\end{cases}
\end{aligned}
\end{equation*}
\end{lemma}

\begin{theorem}\label{thm:1201}
Let $n=Q^m-1$, $Q=q^2$, $m\geq3$, and $\mathcal{C'_\delta}$ be given in (\ref{eq:1125}).
Then $\mathcal{C'_\delta}$ is a linear code with parameters $[2n,n,d(\mathcal{C'_\delta})]$, where
\begin{equation*}
\begin{aligned}
d(\mathcal{C'_\delta})\geq\begin{cases}
q^{m+1}-aq-s+1, &\text{if $m$ is even, $q\leq\sqrt{2s+a^2}+a$, $1\leq a \leq q-2$, $0\leq s \leq q-1$}\\&\text{and $(aq+s)q^{m-1}\leq \delta \leq (aq+s+1)q^{m-1}-1$},\\
q^{m+1}-q^2+q-s+1, &\text{if $m$ is even,  $0\leq s\leq q-2$} \\&\text{and $(q^2-q+s)q^{m-1}\leq\delta\leq(q^2-q+s+1)q^{m-1}-1$},\\
q^{m+1}-q^2+2, &\text{if  $m$ is even  and $(q^2-1)q^{m-1}\leq \delta \leq q^{m+1}-q^2+1$},
\end{cases}
\end{aligned}
\end{equation*}
 and
$$d(\mathcal{C'_\delta})\geq  2\delta$$ if one of the following statements hold:
\begin{itemize}
\item[(1)] if $m$ is odd, $\delta=q^m-b$ and $1\leq b \leq q^2 -2$;
\item[(2)] if $m$ is even, $q>\sqrt{2s+a^2}+a$, $1\leq a \leq q-2$, $0\leq s \leq q-1$ and $(aq+s)q^{m-1}\leq \delta \leq (aq+s+1)q^{m-1}-1$;
\item[(3)]  if  $m $ is even, $sq^{m-1}\leq \delta \leq (s+1)q^{m-1}-1$ and $\frac{\sqrt{2}}{2}q < s\leq q-1;$
\item[(4)] if $m$ is odd, $\frac{\sqrt{2}q^2-2s}{2q}\leq a \leq q-2$, $0\leq s \leq q-1$ and $(aq+s)q^{m-2}\leq \delta \leq (aq+s+1)q^{m-2}-1$;
\item[(5)]  if $m$ is odd,  $0\leq s\leq q-2$, $q\geq4$  and $(q^2-q+s)q^{m-2}\leq\delta\leq(q^2-q+s+1)q^{m-2}-1$;
\item[(6)]   if  $m$ is odd and $(q^2-1)q^{m-2}\leq \delta \leq q^{m}-q^2+1$.
\end{itemize}
All the given $d(\mathcal{C'_\delta})$ mentioned above are better than the square-root lower bound.
\end{theorem}

\begin{proof}  We only prove $d(\mathcal{C'_\delta})$ for the case that $(aq+s)q^{2t-1}\leq \delta \leq (aq+s+1)q^{2t-1}-1$ for $1\leq t \leq m-1$, $1\leq a \leq q-2$ and $0\leq s \leq q-1$. The other cases are similar and we omit the details.

From Lemma \ref{lem:19} and the BCH bound,  we know that
$$d(\mathcal{C}_{\delta})\geq (aq+s)q^{2t-1}\,\, \text{and}\,\, d(\mathcal{C}_{\delta}^{\perp H})\geq q^{2(m-t)+1}-aq-s+1.$$
It is easy to see that $t=\frac{m}{2}$ for $m$ being even and $t=\frac{m-1}{2}$ for $m$ being odd if $\min\{2d(\mathcal{C}_{\delta}),d(\mathcal{C}_{\delta}^{\perp H})\}\geq \sqrt{2n}$.

If $t=\frac{m}{2}$, it is clear that $\min\{2d(C_{\delta}),d(C_{\delta}^{\perp H})\}\geq \sqrt{2n}$ for $1\leq a \leq q-2$ and $0\leq s \leq q-1$.
In addition
{ $$2aq^m+2sq^{m-1} > q^{m+1}-aq-s+1,$$}
which is the same as
$$q^2-2s-2aq\leq0,$$
which is the same as
$$q\leq\sqrt{2s+a^2}+a.$$
Hence, from (\ref{eq:1127}) we know that
\begin{equation*}
\begin{aligned}
d(\mathcal{C'_\delta})\geq\begin{cases}
q^{m+1}-aq-s+1, \,\,& \text{if $q\leq\sqrt{2s+a^2}+a$,}\\
2\delta,\,\,& \text{if $q>\sqrt{2s+a^2}+a$}.
\end{cases}
\end{aligned}
\end{equation*}
 This completes the proof.
\end{proof}

\begin{theorem}\label{thm:1202}
Let $n=Q^m-1$ and $m\geq3$. Let $\mathcal{C'_\delta}$ be given in (\ref{eq:112501}).
Then $\mathcal{C'_\delta}$ is a linear code with parameters $[2n,n,d(\mathcal{C'_\delta})]$, where
\begin{equation*}
\begin{aligned}
d(\mathcal{C'_\delta})\geq\begin{cases}
 2q^m-2s+2, &\text{{ if  $m$ is odd,  $2\leq s \leq q-1$, $sq^m\leq \delta \leq (s+1)q^m-1$ and}} \\
 &\text{ $(aq+s)q^{m-1} \leq \delta \leq (aq+s+1)q^{m-1}-1$},\\
 2q^m-2aq-2s+2,&\text{ if $m$ is odd, $1\leq a \leq q-2$, $0\leq s \leq q-1$ and}\\
 &\text{$(aq+s)q^m \leq \delta \leq (aq+s+1)q^m-1$},\\
2q^m-2q^2+2q-2s+2, &\text{ if  $m$ is odd, $0\leq s \leq q-2$ and}\\
 & \text{$(q^2-q+s)q^m\leq\delta\leq(q^2-q+s+1)q^m-1$},\\
2q^m-2q^2+4, &\text{if $m$ is odd and $(q^2-1)q^m\leq \delta \leq q^{m+2}-q^2+1$},\\
\end{cases}
\end{aligned}
\end{equation*}
 and
$$d(\mathcal{C'_\delta})\geq  \delta$$ if one of the following statements hold:
\begin{itemize}
\item[(1)] if  $m$ is even,  $2\leq a \leq q-2$, $0\leq s \leq q-1$, $q^2\geq aq+s$ and  $(aq+s)q^{m-1} \leq \delta \leq (aq+s+1)q^{m-1}-1$;
\item[(2)] if $m$ is even, $0\leq s \leq q-2$ and  $(q^2-q+s)q^{m-1}\leq\delta\leq(q^2-q+s+1)q^{m-1}-1$;
\item[(3)] if $m$ is even and $(q^2-1)q^{m-1}\leq \delta \leq q^{m+1}-q^2+1$.
\end{itemize}
All $d(\mathcal{C})$ mentioned above are better than the square-root lower bound.
\end{theorem}

\begin{proof} Similar as Theorem \ref{thm:1201}, we  only prove $d(\mathcal{C'_\delta})$ for the case that
 $(aq+s)q^{2t-1}\leq \delta \leq (aq+s+1)q^{2t-1}-1$, where $1\leq t \leq m-1$, $1\leq a \leq q-2$ and $0\leq s \leq q-1$.
In this case, from Lemma \ref{lem:19} and the BCH bound, we have
$$ d(\mathcal{C}_{\delta})\geq (aq+s)q^{2t-1}\,\, \text{and}\,\, d(\mathcal{C}_{\delta}^{\perp H})\geq q^{2(m-t)+1}-aq-s+1.$$
It is easy to see that $t=\frac{m}{2}$ for $m$ being even and $t=\frac{m+1}{2}$ for $m$ being odd if $\min\{d(\mathcal{C}_{\delta}),2d(\mathcal{C}_{\delta}^{\perp H})\}\geq \sqrt{2n}$.

If $t=\frac{m}{2}$, from the range of $a$ and $s$ above, it is easy to see that
$$aq^m+sq^{m-1} < 2q^{m+1}-2aq-2s+2.$$
Then
$$d(\mathcal{C'_\delta})\geq \min\{d(\mathcal{C}_{\delta}),2d(\mathcal{C}_{\delta}^{\perp H})\}=d(\mathcal{C}_{\delta})> \sqrt{2n}$$
 for $2\leq a \leq q-2$ and $0\leq s \leq q-1$.

 If $t=\frac{m+1}{2}$, we have
$$ d(\mathcal{C}_{\delta}^{\perp H})\geq q^{m}-aq-s+1\,\, \text{and}\,\, d(C_{\delta})\geq(aq+s)q^{m}.$$
Obviously,  $\min\{d(\mathcal{C}_{\delta}),2d(\mathcal{C}_{\delta}^{\perp H})\}=d(\mathcal{C}_{\delta})> \sqrt{2n}$ for $1\leq a \leq q-2$ and $0\leq s \leq q-1$.
 This completes the proof.
\end{proof}

In the following, we present two  classes of linear codes of length $2n=\frac{2(Q^m-1)}{Q^s-1}$ with dimension $n$ and minimum distance better than like-square-root lower bound ($d= \sqrt{\frac{n}{2}}$). From Theorem \ref{eq:0822}, we have the following result. The proofs are similar as Theorem \ref{thm:1201} and Theorem \ref{thm:1202}, we omit the details.

\begin{theorem}\label{thm:21}
Let $n=\frac{Q^m-1}{Q-1}$, $m\geq3$ and $\frac{Q^{\frac{m+1}{2}}-1}{Q-1} < \delta \leq \frac{Q^{\frac{m+1}{2}}-1}{Q-1}+q^{\frac{m-1}{2}}$. Then $\mathcal{\mathcal{C'_\delta}}$ is a a linear code with parameters $[2n,n, d(\mathcal{C'_\delta})]$, where
\begin{equation*}
\begin{aligned}
d(\mathcal{C'_\delta})\geq\begin{cases}
 Q^{\frac{m-1}{2}}+1, &\text{ if  $\mathcal{C'_\delta}$ is given in (\ref{eq:1125}),} \\
\delta, &\text{ if  $\mathcal{C'_\delta}$ is given in (\ref{eq:112501}),}
\end{cases}
\end{aligned}
\end{equation*}
which is larger than like-square-root lower bound.
\end{theorem}

\subsection{Self-dual codes with  square-root lower bound }

In this subsection, we construct several classes of Euclidean and Hermitian self-dual linear codes.
The following is a key proposition for us to get these codes.

\begin{proposition}\label{thm}
Let $\cD$ be a linear code over $\mathbb{F}_q$ with parameters $[n,k,d_1]$.  Let $A=\left( \begin{array}{cccccc}
 a& b\\
 c& d\\
\end{array} \right)\in \mathbb M_{2\times 2}(\mathbb{F}_q)$ with $ad\neq bc$ and  $\C=[\cD, \cD^{\bot}]A$. Then $\C$ is an $[2n,n, d]$ self-dual code if and only if one of the following holds

\begin{itemize}
\item[(1)] $2\mid q$,  $a=b$, $c\neq d$ and $\cD$ is dual-containing.

\item[(2)] $2\mid q$, $c=d$, $a\neq b$ and $\cD$ is self-orthogonal.

\item[(3)] $q\equiv 3\pmod 4$ and $\cD$ is self-dual.

\item[(4)] $q\equiv 1\pmod 4$, $\mu^2=-1$, $a=\mu b$ and $c=-\mu d$, or $a=-\mu b$ and $c=\mu d$.

\item[(5)] $q\equiv 1\pmod 4$, $\mu^2=-1$, $a=\pm \mu b$, $c\neq \pm \mu d$ and $\cD$ is dual-containing.

\item[(6)] $q\equiv 1\pmod 4$, $\mu^2=-1$, $a\neq\pm\mu b$, $c= \pm\mu d$ and $\cD$ is self-orthogonal.

\item[(7)] $q\equiv 1\pmod 4$, $\mu^2=-1$, $a\neq\pm\mu b$, $c\neq \pm\mu d$ and $\cD$ is self-dual.
\end{itemize}
Moreover, $d\ge \mbox{min}\{2d_1, d_1^{\bot}\}$ and the equality holds when $A$ is also triangle, where $d_1^{\bot}$ is the minimum distance of the dual code of $\cD$.
\end{proposition}

\begin{proof} Let the generator and parity-check matrices of $\cD$ are $G$ and $H$, then the generator matrix of $\C$ is $\left( \begin{array}{cccccc}
 aG& bG\\
 cH& dH\\
\end{array} \right)$. Hence, $\C$ is self-orthogonal if and only if

\begin{eqnarray*}&&\left( \begin{array}{cccccc}
 aG& bG\\
 cH& dH\\
\end{array} \right)\left( \begin{array}{cccccc}
 aG& bG\\
 cH& dH\\
\end{array} \right)^T\nonumber\\
&=& \left( \begin{array}{cccccc}
 aG& bG\\
 cH& dH\\
\end{array} \right)\left( \begin{array}{cccccc}
 aG^T& cH^T\\
 bG^T& eH^T\\
\end{array} \right)\nonumber\\
&=& \left( \begin{array}{cccccc}
 (a^2+b^2)GG^T& (ac+bd)GH^T\\
 (ac+bd)G^TH& (c^2+d^2)HH^T\\
\end{array} \right)=\mathbf{0}.
\end{eqnarray*}
Notice that $GH^T=\mathbf{0}$ and $HG^T=\mathbf{0}$. Due to the dimension, therefore, $\C$ is self-dual if and only if $ (a^2+b^2)GG^T=(c^2+d^2)HH^T=\mathbf{0}$.

We know that $\cD$ is self-orthogonal if and only if $GG^T=\mathbf{0}$ and   if and only if $HH^T=\mathbf{0}$. We divide the proof into the following parts:

\begin{itemize}
\item[$\bullet$] If $a^2+b^2=0$ and $c^2+d^2=0$, then $a=b$ and $c=d$ when $2\mid q$ and
$a^2=-b^2$ and $c^2=-d^2$ when $-1$ is square in $ \mathbb{F}_q^*$, namely  $q\equiv1\pmod 4$. Assume that $\lambda^2=-1$ for some $\lambda\in \mathbb{F}_q$. Then $(a-b\mu)(a+b\mu)=0$ and $(c-d \mu)(c+d\mu)=0$. From $ad\neq bc$, we have $a=b\mu$, $c=-d\mu$ or $a=-b\mu$, $c=d\mu$.

\item[$\bullet$] If $a^2+b^2=0$ and $c^2+d^2\neq 0$, then we should have $HH^T=\mathbf{0}$, which means  $\cD$ is dual-containing.

\item[$\bullet$]  If $a^2+b^2\neq 0$ and $c^2+d^2= 0$, then we should have $GG^T=\mathbf{0}$, which means $\cD$ is self-orthogonal.

\item[$\bullet$]  If $a^2+b^2\neq0$ and $c^2+d^2\neq 0$, then we should have $HH^T=\mathbf{0}$ and $GG^T=\mathbf{0}$, which means $\cD$ is self-dual.

\end{itemize}
From Lemma \ref{lem:mpcode} and the above discussion, the desired conclusion follows.
\end{proof}

From Subsection B and Proposition \ref{thm}, we can obtain several classes  Euclidean and Hermitian self-dual linear codes with minimum distances better than square-root or square-root-like lower bounds. We now give some remarks about our results.

\begin{remark}\label{rem:123}
From (4) in Proposition \ref{thm}, we know that there is a nonsigular matrix $A$ such that the linear codes in Theorems \ref{eq:thm9190}-\ref{thm:17} or Theorems \ref{thm:1201}-\ref{thm:21}  are  Euclidean or Hermitian self-dual with minimum distances better than square-root or square-root-like lower bounds if $q\equiv 1 \pmod 4$ or $Q\equiv 1 \pmod 4$.
\end{remark}

\begin{remark}\label{rem:1234}
From (1) in Proposition \ref{thm}, when the designed distances of BCH codes lies within the intersection of (\ref{eq:1209}) and the theorems in Remark \ref{rem:123}, we know that there is a nonsigular matrix $A$ such that the linear codes in the theorems of Remark \ref{rem:123} are  Euclidean or Hermitian self-dual with minimum distances better than square-root or square-root-like lower bounds if $2\,|\, q$ or $2\,|\, Q$.
\end{remark}

In the following,  we will compare some of our results with those presented in \cite{Chen23}.
 Let $\mathbf{C}$ be a narrow-sense BCH code of length $n$ with designed distance
\begin{equation}\label{eq:ss}
\delta= \left\lfloor \frac{q^{\frac{m+1}{2}}-q+1}{\nu}\right\rfloor,
\end{equation}
where $n=\frac{q^m-1}{\nu}>4$,  $m$ is odd and $\nu$ is a proper divisor of $q^m-1$.
Let
\begin{equation}\label{eq:ss02}
\mathbf{C}_1=(\mathbf{C},\mathbf{C}^{\perp})\left( \begin{array}{cccccc}
1& \mu\\
0&1\\
\end{array} \right)
\end{equation}
and
\begin{equation}\label{eq:ss01}
\mathbf{C}_2=(\mathbf{C},\mathbf{C}^{\perp})\left( \begin{array}{cccccc}
1& 1\\
0&1\\
\end{array} \right),
\end{equation}
where $\mu^2=-1$.
 We have the following remarks.

 \begin{remark}\label{rem:12011}
Let $q\equiv 1 \pmod 4$, the authors in \cite[Theorem 10]{Chen23} proved that $\mathbf{C}_1$ is a Euclidean self-dual linear code with parameters $[2n,n, \geq \delta]$, where $\delta$ is given in (\ref{eq:ss}). If $\nu \, | \, q-1$ and $\nu\neq q-1$, from Lemma \ref{lem:1122}, Theorem \ref{eq:thm9190} and Proposition \ref{thm}, we know that  $\mathbf{C}_1$ is a Euclidean self-dual linear code with parameters $[2n,n, \geq \frac{q^{\frac{m+1}{2}}-q}{\lambda}+2]$. Hence, in this case, our lower bound is larger than the bound in \cite[Theorem 10]{Chen23}.
\end{remark}

 \begin{remark}\label{rem:1201101}
Let $q=2^s$ with $s\geq 1$. Let $g(x)$ be the generator polynomial of $\mathbf{C}$ and $h^{\perp}(x)$ be the reciprocal polynomial of $(x^n+1)/g(x)$. Let $\mathbf{C'}$ denote the cyclic code of length $2n$ over $\mathbb{F}_q$ with generator polynomial $g(x)h^{\perp}(x)$. The authors in \cite{Chen23} proved that $\mathbf{C'}$ is a Euclidean self-dual cyclic code with parameters $[2n,n, \geq \delta]$ and permutation-equivalent to $\mathbf{C}_2$, where $\delta$ is given in (\ref{eq:ss}). If $\nu \, | \, q-1$ and $\nu\neq q-1$, from Lemma \ref{lem:1122}, Theorem \ref{eq:thm9190} and Proposition \ref{thm}, we know that  $\mathbf{C'}$ is a Euclidean self-dual cyclic code with parameters $[2n,n, \geq \frac{q^{\frac{m+1}{2}}-q}{\lambda}+2]$. Hence, in this case, our lower bound is larger than the bound in \cite[Theorem 7]{Chen23}.
\end{remark}

\begin{remark}\label{rem:1201102}
Let $n=\frac{Q^m-1}{\mu}>4$, where $q=2^s$ with $s\geq 1$, $m$ is odd and $\mu$ is a proper divisor of $q^m-1$. Let $\mathbf{C}$ be a narrow-sense BCH code of length $n$ with designed distance
$$\delta=  \frac{q^m-1}{\mu}.$$
Let $g(x)$ be the generator polynomial of $\mathbf{C}$ and $h^{\perp H}(x)$ be the generator polynomial of the Hermitian dual of $\mathbf{C}^{\perp H}$.
 Let $\mathbf{C''}$ denote the cyclic code of length $2n$ over $\mathbb{F}_Q$ with generator polynomial $g(x)h^{\perp H}(x)$. The authors in \cite{Chen23} proved that $\mathbf{C''}$ is a Hermitian self-dual cyclic code with parameters $[2n,n, \geq \delta]$ and permutation-equivalent to
$$(\mathbf{C},\mathbf{C}^{\perp})\left( \begin{array}{cccccc}
1& 1\\
0&1\\
\end{array} \right).$$
If $\mu=1$, from Lemma \ref{lem:1122} and Lemma \ref{thm}, we know that  $\mathbf{C''}$ is a Euclidean self-dual cyclic code with parameters $[2n,n, \geq 2\delta]$. Hence, in this case, our lower bound  is $2$ times that of the existing lower bound in \cite[Theorem 8]{Chen23}.
\end{remark}

\section{Concluding remarks }\label{sec:concluding}

The contributions of this paper are summarized as follows:

\begin{itemize}
\item[(1)]  Let $n=\frac{q^m-1}{q^s-1}$, $2\leq t\leq \frac{m}{s}-1$ and $\frac{q^{st}-1}{q^s-1}< \delta \leq \frac{q^{st}-1}{q^s-1}+q^{t-1}$, the authors in \cite{GDL21} showed that $d(\C_{\delta}^{\perp})\geq \frac{q^{m-ts}-1}{q^s-1}+1$. In Theorem \ref{eq:1204}, we presented a new lower bound $d(\C_{\delta}^{\perp})\geq q^{m-st}$, which is  almost $q^s-1$ times that of the existing lower bounds.

\item[(2)]  Let $n=\frac{Q^m-1}{Q-1}$, $2\leq t\leq m-1$ and $\frac{Q^t-1}{Q-1}< \delta \leq \frac{Q^t-1}{Q-1}+q^{t-1}$, the authors in \cite{Fan23} showed that $d(\C_{\delta}^{\perp})\geq\frac{Q^{m-t}q-q}{Q-1}+1$. In Theorem \ref{eq:0822}, we showed a new lower bound $d(\C_{\delta}^{\perp})\geq Q^{m-t}+1$, which is  almost $q$ times that of the existing lower bounds.

\item[(3)] Let  $n=\frac{q^m-1}{\lambda}$, or $n=q^m+1$, or $n=Q^m-1$, where $\lambda \, | \, q-1$ and $\lambda\neq q-1$.  We constructed several classes Euclidean  and Hermitian linear codes with parameters $[2n,n]$ (see Theorem \ref{eq:thm9190}, Theorem \ref{eq:1thm91190}, Remark \ref{eq:1212}, Theorem \ref{thm:16}, Theorem \ref{thm:1201} and Theorem \ref{thm:1202}).  In many cases, these codes are  Euclidean  and Hermitian self-dual codes with minimum distance larger than square-root bounds (see Remark \ref{rem:123} and Remark \ref{rem:1234}).

\item[(4)] Let  $n=\frac{q^m-1}{q^s-1}$, or $n=\frac{Q^m-1}{Q-1}$, where $s\, |\,m$.  We obtained several classes Euclidean  and Hermitian linear codes with parameters $[2n,n]$ (see Theorem \ref{thm:17} and Theorem \ref{thm:21}). In many cases, these codes are  Euclidean  and Hermitian self-dual codes with minimum distance larger than square-root-like bounds (see Remark \ref{rem:123} and Remark \ref{rem:1234}).

\item[(5)] We compared some of our results with those presented in \cite{Chen23}. Our lower bounds are better than the results in \cite{Chen23} (see Remark \ref{rem:12011}, Remark \ref{rem:1201101} and Remark \ref{rem:1201102}).
\end{itemize}

It is hard to construct an infinite family of self-dual cyclic codes with
the square-root lower bound.   We welcome interested readers to join us in exploring this topic.

\end{document}